\documentclass[runningheads,envcountsame,envcountsect]{llncs}

\usepackage[T1]{fontenc}

\title{Computing Fixpoints of Learned Functions: Chaotic Iteration and
  Simple Stochastic Games}

\titlerunning{Computing Fixpoints of Learned Functions} %

\author{Paolo Baldan\inst{1}\orcidlink{0000-0001-9357-5599} \and
  Sebastian Gurke\inst{2}\orcidlink{0009-0008-4343-1384} \and Barbara
  K\"onig\inst{2}\orcidlink{0000-0002-4193-2889} \and Florian
  Wittbold\inst{2}\orcidlink{0000-0001-8307-503X}}

\institute{Department of Mathematics, University of Padua, Italy
  \\\email{baldan@math.unipd.it} \and Universit\"at
  Duisburg-Essen\\\email{\{sebastian.gurke,barbara\_koenig,florian.wittbold\}@uni-due.de}}

\authorrunning{P. Baldan, S. Gurke, K\"onig and F. Wittbold} %

\usepackage[utf8]{inputenc}
\usepackage{amsmath}
\usepackage{amssymb}
\usepackage{nicefrac}
\usepackage{hyperref}
\usepackage{cleveref}
\usepackage{mathtools}
\usepackage{csquotes}  
\usepackage{caption}
\usepackage{subcaption}

\usepackage[appendix=append,bibliography=common]{apxproof}

\usepackage{stmaryrd}
\usepackage{orcidlink}
\usepackage{tikz,pgfplots}

\newtheoremrep{theorem}{Theorem}[section]
\newtheoremrep{maintheorem}[theorem]{Main Theorem}
\newtheoremrep{corollary}[theorem]{Corollary}
\newtheoremrep{lemma}[theorem]{Lemma}

\crefname{maintheorem}{theorem}{theorems}
\Crefname{maintheorem}{Theorem}{Theorems}

\newcommand{\R}{\mathbb{R}} %
\newcommand{\Rp}{\R_+}
\newcommand{\Ri}{\overline{\R}}
\newcommand{\Rpi}{\overline{\R}_+}
\newcommand{\N}{\mathbb{N}}

\newcommand{\supp}[1]{\mathrm{supp}\,#1}

\newcommand{\reg}{progressing}
\newcommand{\Reg}{Progressing}

\newcommand{\fix}[1]{\mathrm{Fix}(#1)}
\newcommand{\omegaclos}[1]{\overline{#1}}
\newcommand{\lr}{\alpha}
\newcommand{\df}{\beta}

\newcommand{\mystrut}{\raisebox{2ex}[0.4cm]{}}

\newcommand{\cut}[1]{}

\begin{document}

\maketitle

\begin{abstract}
  The problem of determining the (least) fixpoint of
  (higher-dimensional) functions over the non-negative reals
  frequently occurs when dealing with systems endowed with a
  quantitative semantics.
  We focus on the situation in which the functions of interest
  are not known precisely but can only be approximated.
  As a first contribution we generalize an iteration scheme called
  dampened Mann iteration, recently introduced in the literature.
  The improved scheme relaxes previous constraints on parameter
  sequences, allowing learning rates to converge to zero or not
  converge at all.
  While seemingly minor, this flexibility is essential to enable
  the implementation of chaotic iterations, where only a subset of
  components is updated in each step, allowing to tackle
  higher-dimensional problems.
  Additionally, by allowing learning rates to converge to zero, we can
  relax conditions on the convergence speed of function
  approximations, making the method more adaptable to various
  scenarios.
  We also show that dampened Mann iteration applies immediately to
  compute the expected payoff in various probabilistic models,
  including simple stochastic games, not covered by 
  previous work.
\end{abstract}

\section{Introduction}

This paper focuses on the problem of identifying fixpoint iteration
schemes
for determining the (least) fixpoint of a
$d$-dimensional function $f$ on the reals in scenarios where this
function is not known precisely, but can only be
approximated.
Concretely, we assume to be able to construct a
sequence of %
approximating functions
$f_1,f_2,f_3,\dots$ that converges to $f$.
The problem is non-trivial for a number of reasons, starting from the
fact that the least fixpoint operator is not necessarily continuous
and hence the sequence of least fixpoints of the functions $f_n$ does
not always converge to the least fixpoint of~$f$.

This task can be seen as an abstract formulation of the problem that
occurs in reinforcement learning, where a Markov Decision Process
(MDP) -- a transition system based on probabilistic and
non-deterministic branching as well as rewards -- is given and the aim
is to determine the expected return and the optimal strategy
to achieve it. In a real-world scenario the probabilities
(and possibly also the rewards) of such MDPs are often not known, but
can only be determined by exploring the MDP and
sampling. Reinforcement learning algorithms such as
Q-learning~\cite{WD:QL}, SARSA~\cite{r:sarsa} or
Dyna-Q~\cite{s:dyna-integrated-architecture,s:planning-incremental-dyn-prog,kaelbling1996reinforcement}
typically proceed by interleaving exploration of the MDP and strategy
computation.
These algorithms can be classified as model-based (e.g., Dyna-Q) --
where the (approximated) MDP model is explicitly computed -- and
model-free (e.g., SARSA, Q-learning) -- where updates are not dependent
on the estimated probabilities, but on the current sample.

Indeed, the expected return of a (finite state) MDP can be
characterized as the least fixpoint of a function based on the Bellman
equation.
This function is monotone with respect to the pointwise order on tuples
and non-expansive ($1$-Lipschitz) with respect to the supremum metric,
i.e., function application does not increase the distance of two points
in the space.
The fact that the same happens for a number of other fixpoint equations
of interest (e.g., for computing payoffs for simple stochastic
games~\cite{c:complexity-stochastic-game} and for behavioural
metrics~\cite{bblm:on-the-fly-exact-journal}) leads us to
develop our theory for monotone and non-expansive functions.

In~\cite{BGKPW:CAFDM}, we proposed a solution to this problem:
instead of computing fixpoints via Kleene iteration -- as 
common for such models -- the idea is to perform a so-called dampened
Mann iteration that adds a dampening factor to the well-known
Mann iteration~\cite{b:iterative-approximation-fixed-points}. 
Concretely the following iteration schema is proposed:
\[
  x_{n+1} = (1-\beta_n)\cdot\left(x_n + \lr_n\cdot(f_n(x_n)-x_n)
  \right)%
\]

The parameter $\lr_n$, often referred to as the learning rate, is
used to take a weighted sum between the previous value and the value
obtained by applying the $n$-th approximated function, making the
iteration more robust with respect to perturbations.
A central ingredient is the dampening factor $1-\beta_n$, inspired
by~\cite{KIM200551}. Its purpose is to decrease a value that for some
reasons might over-estimate the least fixpoint. This is essential when the function of interest admits more than one fixpoint
and the iteration in early phases can get stuck at a fixpoint that is
not the least.
According to~\cite{BGKPW:CAFDM}, the parameters $\lr_n,\beta_n$ have
to be chosen such that $\lr_n$ converges to $1$ (the iteration
scheme converges to a Kleene iteration) and $\beta_n$ converges to $0$
(the dampening is vanishing over time). For most results the first
condition can be relaxed to $\lr_n$ converging to some
$\alpha > 0$.

Then, the sequence $x_n$ generated by the above iteration scheme is
guaranteed to converge to the least fixpoint of $f$ in a number of
situations:
when the function $f$ is a (power) contraction,
when the the sequence
$(f_n)$ converges monotonically to $f$ and when $(f_n)$ converges
normally.
The latter means that $\sum_{n} \|f - f_n\| < \infty$ (where
$\|\cdot\|$ is the supremum norm), a condition which can always be
enforced by ``speeding'' up the iteration and determining a function
$g$ such that $f_{g(i)}$ converges fast enough and thus normally.
For instance, consider the sequence of functions
$f_n(x) = \nicefrac{1}{n} + \left(1-\nicefrac{1}{n}\right)\cdot x$ converging to the
identity $f(x) = x$.
It can be seen that $(f_n)$ does not converge normally as
$\|f_n-f\|=\nicefrac{1}{n}$ and, indeed, the dampened Mann
iteration scheme from~\cite{BGKPW:CAFDM} fails.
However, normal convergence can be enforced by considering a
subsequence $f_{g(n)}$ with $g(n) =n^2$.
Whenever approximations are realized by means of samplings,
this speedup can be obtained by increasing the number of
samplings performed between two iterations.
Interestingly, the mentioned paper also shows that, in the case of
MDPs, approximations obtained by sampling always guarantee convergence
of the dampened Mann iteration scheme, independently of the speed of
convergence.

The solution we provided in~\cite{BGKPW:CAFDM} has some drawbacks.
A relevant issue resides in the fact that it implicitly assumes that
all components are updated at each iteration.
Clearly, in applications such as reinforcement learning where large
models can produce functions working in dimensions of several thousands,
this becomes impractical if not impossible.
Chaotic iterations~\cite{FROMMER2000201} instead offer a scalable
alternative where only a subset (possibly only one) of the components
are updated at each iteration.
The components to be updated might be chosen, for example, to be the
ones that seem statistically more relevant to the goal (such as states
that are visited more often in sampled runs of a Markov decision
process) or that are dependent on components that were previously updated.
Using chaotic strategies can often result in less computation and while
not necessarily so, its faster update cycle can give better insight
into the convergence behaviour at runtime or even produce approximations
where full updates are infeasible.

While this might look like a minor implementation aspect, one can see
that generalizing to chaotic iteration requires quite a radical change
to the theory.
The parameters $\lr_n$ and $\beta_n$ have to become dependent on the
dimension (the state), i.e., they become themselves vectors.
If we write $z(i)$ for the $i$-th component of a vector $z$,
the iteration schemes becomes of the following form, where setting
$\lr_n(i) =\beta_n(i)=0$ we ensure that component $i$ is not updated
at step $n$:
\[
  x_{n+1}(i) = (1-\beta_n(i))\cdot\left(x_n(i) +
    \lr_n(i)\cdot(f_n(x_n)-x_n)(i) \right)
\]

Since in a proper chaotic iteration, components may not be updated
infinitely often, parameters $\lr_n(i)$ must be allowed to be $0$
infinitely often, which means that they either converge to $0$ or
do not converge at all.
This deeply deviates from the theory in~\cite{BGKPW:CAFDM}, which
heavily relies on the learning rate sequence $\lr_n$ converging to
$1$ or at least to a value strictly larger than $0$.

In this paper, we will indeed show that the iteration scheme can work with more
general parameter sequences, where the learning rate can converge to
$0$ or is even allowed not to converge.
This relaxation of the condition on the parameters requires a substantial
change in the proof techniques.

In addition to implementing chaotic iteration, the additional
flexibility for the parameter sequences can also be used to treat
different components of the system in a different way: possibly we
have rather precise knowledge about the values in one component and
know that we are close to the solution.
Then it makes sense to choose $\beta_n(i)$ and $\lr_n(i)$ close to zero.
On the other hand, it could be the case that in another component we have
little information and still have to iterate more and correct a potential
over-approximation.
Then the parameters should be chosen further away from $0$.

The more general schema allows us to match it with the practice of a
number of model-based reinforcement learning techniques.
For instance, Dyna-Q, a model-based extension of Q-learning, uses
chaotic iteration and considers a learning rate $\lr_n$ converging to $0$.
More generally, the present work might be a useful step also towards the
development of a theory encompassing model-free reinforcement learning
techniques.
In fact, in a model-free approach, such as Q-learning~\cite{WD:QL}, the
update is based on the current sample and not on the estimated probabilities
and rewards. Hence, in order to
converge to the solution, the learning rate $\lr_n$ must converge to
$0$, giving more and more weight to the previous estimate rather than
to the result of applying the function, otherwise the iteration would
oscillate (due to the varying samples) rather than converge.
Existing convergence results (e.g.,~\cite{bt:neuro-dynamic-programming,WD:QL}) for these algorithms typically
assume functions with unique fixpoints (as is the case for discounted
problems or systems with some termination guarantees) or
convergence to some fixpoint, while the case of convergence to the
least fixpoint is not treated.

Finally,  the possibility of setting a learning rate $\lr_n$
 converging to $0$ also offers the potential to relax the
condition on the speed of convergence of the sequence
$f_n$. While we required in~\cite{BGKPW:CAFDM} that $f_n$ converges normally to
$f$, i.e., $\sum_{n} \|f - f_n\| < \infty$, here we can relax the
condition to $\sum_{n} \lr_n\|f - f_n\| < \infty$,
which is easier to satisfy when $\lr_n$ tends to $0$
and thus the iteration relies less and less on the approximations.

Another relevant question -- unanswered in~\cite{BGKPW:CAFDM} -- is
whether dampened Mann iteration works -- out of the box -- for simple
stochastic games (SSGs). The mentioned paper only offers a solution,
sketched above, involving a speedup of the sampling process for ensuring
normal convergence. However, this is costly since it requires to
obtain many samples before being able to do the next iteration and
obtain the next estimate of the solution. Even
worse, the periods until the next estimate are getting longer and
longer since the number of needed samples might increase non-linearly.
The question of
whether one can instead directly apply dampened Mann iteration to
SSGs
was thus relevant and still open.

The results from~\cite{BGKPW:CAFDM} cannot
be straightforwardly be adapted to handle SSGs. The
previous proof strategy was to first show that dampened Mann iteration
works in the exact case and for the case of power contractions. From
there one is able to derive other results, for instance for MDPs
(where the fact is exploited that the Bellman equations of MDPs
without end components result in fixpoint functions that are power
contractions). Including SSGs in the picture and
accomodating for the generalized parameter sequences, requires
completely new proofs. We prove convergence of dampened Mann iteration
for sequences of functions converging monotonically and generalize
the result for power contractions to the case of Picard operators.
Together with the fact
that the least fixpoint operator is continuous for Markov Chains and
for SSGs, these results allow us to conclude that
dampened Mann iteration works for both MDPs and SSGs. For MDPs the novelty lies in the new types of parameters, while
for SSGs this result was completely open
in~\cite{BGKPW:CAFDM}.

Being able to treat
SSGs also ensures that we can
solve (least) fixpoint equations for behavioural
metrics~\cite{b:prob-bisimilarity-distances,bw:complexity-bisimilarity-pseudometric},
another important application
area. In~\cite{kt:approximate-bisimulation-minimization,FKPB:RPBLMC},
it was observed that a small perturbation of the transition
probabilities in a model can drastically change the distance and the
authors suggest an alternative, more robust, notion of distance. 
With our results it is now also possible to stay with the
original notion of pseudometric and use dampened Mann iteration with
increasingly better approximations to converge to the true distance.

In summary, our main contribution is a generalisation of the dampened
Mann iteration scheme proposed in previous work for approximating fixpoints
of functions arising from quantitative models. More precisely

\begin{itemize}
\item We allow the learning rate to converge to $0$ or not
  converge at all, enabling better adaptation to different convergence
  rates of function approximations.

\item We introduce a generalized dampened Mann iteration where the
  learning rate and dampening parameters can vary across
  dimensions, allowing for chaotic iteration strategies.

\item We establish convergence of the introduced schemes for simple
  stochastic games and their sampled approximations.
\end{itemize}

\section{Preliminaries and Notation}
We denote by $\R=(-\infty,\infty)$ the set of real numbers, by $\Rp=[0,\infty)$ the set of
non-negative reals, by $\Ri=[-\infty,\infty]$ and $\Rpi=[0,\infty]$ the sets of (non-negative)
real numbers including infinity, and by $\N_0$ and $\N$ the set of natural numbers
with and without $0$, respectively.

For sets $X,Y$, we will denote by $Y^X$ the set of functions $f\colon X\to Y$.
If $X$ is finite, we will identify $Y^X$ with the set of vectors $Y^{|X|}$. A sequence in $X$ is denoted by $(x_n)_{n \in \N_0}$ or simply by $(x_n)$ when the index is clear from the context.
For $x\in\R^d$ and $i\in\{1,\dots,d\}$, we will write $x(i)\in\R$ to denote its $i$-th component and,
if clear from the context, we will sometimes identify $x\in\R$ with the vector $(x,\dots,x)\in\R^d$.
For $x,y\in\R^d$, $x\cdot y=(x(1)y(1),\dots,x(d)y(d))\in\R^d$ will denote their component-wise 
product.

For $x,y\in\Ri^d$, we write $x\leq y$ for the pointwise (partial) order, i.e.,
$x(i)\leq y(i)$ for all $i\in\{1,\dots,d\}$.
A function $f\in Y^X$, where $X,Y\subseteq\Ri^d$, is monotone if, for all
$x,y\in X$, we have that $x\leq y$ implies $f(x)\leq f(y)$.
It is called $\omega$-continuous if, for all $(x_n)\subseteq X$ with $x_1\leq x_2\leq\dots$ and
$x=\sup_{n\in\N}x_n\in X$, we have $f(x)=\sup_{n\in\N}f(x_n)$.
Note that any $\omega$-continuous function is monotone.

We equip $\R^d$ with the \emph{supremum norm} defined as $\|x\|=\sup\{x(i)\mid i\in\{1,\dots,d\}\}$
for $x\in\R^d$ and extended to functions $f\in Y^X$, where $X,Y\subseteq\R^d$, by
$\|f\|=\sup_{x\in X}\|f(x)\|\in[0,\infty]$.
A function $f\in Y^X$ is $L$-Lipschitz
for a constant $L\in\Rp$ if, for all $x,y\in X$, $\|f(x)-f(y)\|\leq L\|x-y\|$.
It is non-expansive if it is $1$-Lipschitz.
It is a contraction if it is $q$-Lipschitz for contraction factor $q<1$.

A fixpoint of $f\colon X\to X,X\subseteq\Rp^d$, is $x\in X$ such that
$f(x)=x$.  We denote the set of fixpoints of $f$ by $\fix{f}$ and, in
case it exists, the \emph{least fixpoint} of $f$ by $\mu f$.  We will
be interested in approximating least fixpoints of non-expansive and
monotone functions $f\colon X\to X$ where $X$ is a \emph{$0$-box},
i.e.\ a set of the form
\begin{equation}
  \label{eq:domain-form}
  X=\{x\in\Rp^d\mid x\leq x^*\} \qquad \qquad \mbox{for some $x^*\in\Rpi^d$.}
\end{equation}

Each such function $f$ extends  to an $\omega$-continuous function
$\omegaclos{f}\colon\Rpi^d\to\Rpi^d$ given by $\omegaclos{f}(x)=\sup\{f(y)\mid y\leq x\}$.

Using Kleene's fixpoint theorem, $\omegaclos{f}$ has a least fixpoint and
$\fix{f}\neq\emptyset$ iff $\mu\omegaclos{f}\in X$ in which case
$\mu f=\mu\omegaclos{f}$.  By abusing the notation, we will
simply write $\mu f$ instead of $\mu\omegaclos{f}\in\Rpi^d$.

A function $f\colon X \to X$ is a power contraction if there is some $k\in\mathbb{N}$ such that the $k$-fold iteration $f^k$ is a contraction.

The map $f$ is called a Picard operator if it
has a unique fixpoint $x^*\in X$ and the sequence $(f^n(x))$
converges to $x^*$ for every $x\in X$.
If $f\colon X\to X$ is a contraction on a closed (hence complete) set
$X\subseteq\Rp^d$, then, by Banach's theorem, it is a Picard operator.

\paragraph{Dampened Mann iteration.}

For the rest of the paper, we will be interested in approximating
least fixpoints of non-expansive and monotone functions
$f\colon X\to X$, where $X$ is a $0$-box, given only approximations
$f_n\colon X\to X$ of $f$.  In~\cite{BGKPW:CAFDM}, we suggested the
use of a dampened Mann iteration, which is a variation of Mann
iteration scheme with the addition of dampening factor.

\begin{definition}[Mann scheme]
  \label{de:dampened-mann}
  A \emph{(dampened) Mann scheme} $\mathcal{S}$ is a pair
  $((\lr_n),(\beta_n))$ of parameter sequences in $[0,1)$,
  referred to as \emph{learning rates} and \emph{dampening factors},
  respectively, such that
  \begin{align}
    \label{eq:beta-condition-1}
    \lim_{n \to \infty} \beta_n &= 0, \\
    \label{eq:beta-condition-2}
    \sum_{n \in \mathbb{N}_0} \beta_n &= \infty 
    \quad \text{(or equivalently, $\prod_{n \in \mathbb{N}_0} (1 - \beta_n) = 0$)}.
  \end{align}
  Given a $0$-box $X\subseteq\Rp^d$ and a sequence
  $(f_n)$ of functions $f_n\colon X\to X$,
  a Mann scheme $\mathcal{S}$ generates a class of sequences $(x_n)$ with $x_0 \in X$  chosen aritrarily and 
  \begin{equation}
    \label{eq:iteration} 
      x_{n+1} = (1-\beta_n)\cdot\left(x_n + \lr_n\cdot(f_n(x_n)-x_n) \right)
  \end{equation}
  
  A Mann scheme $\mathcal{S}$ is  \emph{exact} if for all monotone and
  non-expansive $f\colon X\to X$,
  setting $f_n=f$ for all $n$, the sequences $(x_n)$ generated by
  $\mathcal{S}$ converge to $\mu f$.
\end{definition}

Intuitively, the learning rates determine to what extent the update
$f(x_n)-x_n$ is applied to the current guess $x_n$, ranging from a
full Kleene update ($\lr_n=1$, implying $x_{n+1} = (1-\beta_n) f_n(x_n)$) to no update at all ($\lr_n=0$, implying $x_{n+1}=(1-\beta_n) x_n$).
The dampening factors ensure that over-approximations of $\mu f$ introduced by starting from $x_0$ instead of $0$ and using $f_n$ instead of $f$, and which
would lead to possible divergence of a Kleene iteration, can be
dampened. Condition~\eqref{eq:beta-condition-1} ensures that dampening
eventually reduces -- meaning that, in the long run, when we are close
to the least fixpoint of $f$, we stay close to it.  On the other
hand, Condition~\eqref{eq:beta-condition-2} guarantees that at any
stage there is still enough ``dampening power'' left to correct
possible over-approximations.

In~\cite{BGKPW:CAFDM}, we showed several convergence results for
\emph{Mann-Kleene schemes}, i.e., Mann schemes such that
$\lim_{n\to\infty} \lr_n = 1$ and \emph{relaxed Mann-Kleene
  schemes} where $\lim_{n\to\infty} \lr_n > 0$.\footnote{Note that in~\cite{BGKPW:CAFDM} $\alpha_n$ is replaced with $1-\alpha_n$.
  Here, we chose this notation instead as it better reflects the widespread interpretation of $\alpha_n$
  as a learning rate.}
More precisely, it shows that (relaxed) Mann-Kleene schemes are exact and it proves convergence in the ``approximated case'', i.e., when we use at each iteration approximations from a sequence $(f_n)$ converging to $f$, with
conditions on the target functions $f$ (power contraction) or on the mode
of convergence of the sequence $(f_n)$  (monotone or normal convergence).

\section{Generalizing Dampened Mann Iteration}
\label{sec:dmi}

In this section we generalize the convergence results for Mann schemes
from~\cite{BGKPW:CAFDM}. First, in \S\ref{sec:decreasing-lrs} we
overcome the restriction to Mann-Kleene schemes,
and show
convergence results for learning rates which are allowed to converge
to $0$ or to not converge at all. Relying on this, in
\S\ref{ss:gen-mann} we further generalize the scheme by allowing
the parameter values $\lr_n,\beta_n$ to depend on the dimension,
thus, in particular, enabling a chaotic iteration. While
the results for chaotic iteration in \S\ref{ss:gen-mann} are
for the exact case, i.e., the function $f$ is known and used at each
iteration, the extension to the case in which $f$ can be only
approximated via a sequence $(f_n)$ is described in
\S\ref{sec:approx}.

\subsection{Schemes With Non-Converging Learning Rates}
\label{sec:decreasing-lrs}

We start by lifting the restriction to \emph{(relaxed) Mann-Kleene}
schemes, showing that convergence can be guaranteed also when the
learning rate does not converge or converges to zero. In addition,
while~\cite{BGKPW:CAFDM}
restricts to functions
$f\colon X\to X$ with $\mu f<\infty$, here we allow the fixpoint to be
infinite in some (or all) dimensions.

\begin{definition}[{\Reg} scheme]
  A Mann scheme
  $\mathcal{S}=((\lr_n),(\df_n))$ is called
  \emph{\reg}
  if, assuming $\nicefrac{0}{0}=0$ and $\nicefrac{x}{0}=\infty$ for $x>0$, it holds
  \begin{equation}
    \label{ass:quotient}
    \lim_{n\to\infty}\nicefrac{\beta_n}{\lr_n} = 0.
  \end{equation}
\end{definition}

Intuitively, Condition~\eqref{ass:quotient} ensures that the scheme eventually prioritizes updates
over dampening.

Observe that Condition (\ref{eq:beta-condition-1}) of Definition~\ref{de:dampened-mann}, i.e.,
$\lim_{n \to \infty} \beta_n = 0$, is implied by Condition (\ref{ass:quotient}) above, since
$\nicefrac{\beta_n}{\lr_n} \geq \beta_n$.
\cut{Moreover, Conditions~\eqref{ass:quotient} and with Condition~\eqref{eq:beta-condition-2}
($\sum_n \beta_n=\infty$) also implies $\sum_{n \in \N_0} \lr_n = \infty$, since $\lr_n \geq \beta_n$, asymptotically.}
Note also that there is no
 convergence requirement on $\lr_n$ and thus, in particular,
  $\lr_n$ is allowed to tend to $0$.

  Canonical choices of the parameters are $\beta_n=\nicefrac{1}{n}$
  and either $\lr_n=1$ (Kleene updates at every step) or
  $\lr_n=\nicefrac{1}{n^{\varepsilon}}$ with $0<\varepsilon<1$
  (decreasing learning rates).

We first show that, when the exact function $f$ is known and can be used at each iteration, progressing schemes ensure convergence to the least fixpoint.

\begin{theoremrep}[{\Reg} is exact]
  \label{thm:old-exact}
  {\Reg} Mann schemes are exact.
\end{theoremrep}
\begin{proof}
  The theorem follows immediately from the proof of~\Cref{thm:exact}
  for the generalized Dampened Mann iteration.
  \qed
\end{proof}

\begin{toappendix}
It is worth noting that the Theorem~\ref{thm:old-exact}
holds for arbitrary closed and convex domains $X$
with $0\in X$ making this a generalisation of the results
in~\cite{BGKPW:CAFDM}. The proof would remain exactly the same.
In fact, the fact that the restriction of the domain of the functions
to $0$-boxes (cf. Condition~\eqref{eq:domain-form}) is only required
for the generalized iteration used in the sequel to be well-defined
and it is not actually used in any step of the proof.
\end{toappendix}

When $f$ can only be
approximated, the following generalisation of~\cite[Theorem
4.4.2]{BGKPW:CAFDM} shows the benefits of using decreasing
learning rates.

\begin{corollaryrep}[Approximated function]
  \label{cor:normal}
  Let $\mathcal{S}$
  be a {\reg} Mann scheme, $f\colon X\to
  X$ a monotone and non-expansive map over a $0$-box
  $X\subseteq\Rp^d$, and $(f_n)$ a sequence of maps $f_n\colon X\to
  X$ such that $\sum \lr_n \|f - f_n\| < \infty$.
  Then
  the sequences $(x_n)$ generated by $\mathcal{S}$ on $(f_n)$ satisfy $x_n \to \mu f$.
\end{corollaryrep}
\begin{proof}
    We write $(y_n)$ for the sequence generated by~\eqref{eq:iteration} with the exact function,
    i.e., $f_n=f$ for all $n\in\N_0$.

    As in the proof of  in~\cite[Theorem 4.4.2]{BGKPW:CAFDM}, one can show that
    \[\|x_{n+1}-y_{n+1}\|\leq\sum_{i=0}^n\left(\prod_{j=i}^n (1-\beta_j)\right)\lr_i\|f_i-f\|.\]
    Now, the right-hand side converges to $0$: By assumption, for $\varepsilon>0$, there is
    $N\in\N_0$ such that $\sum_{i=N}^\infty\lr_i\|f_i-f\|\leq\nicefrac{\varepsilon}{2}$.
    Choosing $M\in\N_0$ such that
    $\prod_{j=N}^M(1-\beta_j)\leq\varepsilon/(2N\max_{n\in\N_0}\lr_n\|f_n-f\|)$, we get for
    all $n\geq M$ that
    \begin{align*}
        \MoveEqLeft[-1]\sum_{i=0}^n\left(\prod_{j=i}^n (1-\beta_j)\right)\lr_i\|f_i-f\|\\
        &\leq\sum_{i=0}^{N-1}\left(\prod_{j=i}^n (1-\beta_j)\right)\lr_i\|f_i-f\|+\sum_{i=N}^\infty\left(\prod_{j=i}^n (1-\beta_j)\right)\lr_i\|f_i-f\|\\
        &\leq N(\max_{n\in\N_0}\lr_n\|f_n-f\|)(\prod_{j=N}^M(1-\beta_j))+\sum_{i=N}^{\infty}\lr_i\|f_i-f\|\\
        &\leq\nicefrac{\varepsilon}{2}+\nicefrac{\varepsilon}{2}=\varepsilon.
    \end{align*}

    Thus, we have $\|x_n-y_n\|\to 0$ for $n\to\infty$.  As,
    by~\Cref{thm:old-exact}, we have $y_n\to\mu f$, this concludes the
    proof.
    \qed
\end{proof}

This significantly improves~\cite[Theorem 4.6]{BGKPW:CAFDM} which requires
$\sum \|f - f_n\| < \infty$, since when $\lr_n \to 0$, then $\sum \lr_n\|f - f_n\|$ is, of
course, more likely to be bounded. Indeed, for any sequence $(f_n)$ converging to $f$ it can be shown that there is a progressing Mann scheme working for any starting point (see Corollary~\ref{co:alw-work}).

\begin{toappendix} 
  \begin{corollary}
    \label{co:alw-work}
    Let $f\colon X \to X$ be a monotone and non-expansive map and $(f_n)$ be a sequence
    of maps $f_n\colon X \to X$ with $\|f_n - f\| \to 0$.
    Then, there are sequences $(\lr_n)$ and $(\beta_n)$ such that, for every initial value
    $x_0\in X$, the sequence defined by~\eqref{eq:iteration} converges to $\mu f$.
\end{corollary}
\begin{proof}
Define $\epsilon_n\coloneqq\|f_n - f\|$.
We just have to combine the above corollary with the following simple facts from analysis:
\begin{itemize}
\item Since $\epsilon_n \to 0$, there exists a sequence $(\lr_n)$ with $\sum \lr_n = \infty$
    and $\sum \lr_n \epsilon_n < \infty$.
\item Now, there exists a sequence $(\beta_n)$ with $\sum \beta_n = \infty$ and
    $\lim \nicefrac{\beta_n}{\lr_n} = 0$.~\cite{Ash01091997}
\end{itemize}
With these sequences we can apply the above corollary.
\qed
\end{proof}
\end{toappendix}

\begin{example}
  \label{ex:weigthed-convergence}
  Consider the identity $f\colon [0,1] \to [0,1]$ and a sequence
  of approximating functions $f_n\colon [0,1] \to [0,1]$ defined by
  $f_n(x) = (1-\nicefrac{1}{n})x+\nicefrac{1}{n}$.  Then one can see that any sequence $(x_n)$
  generated, starting from $x_0 = 1$, by a relaxed Mann-Kleene scheme with
  dampening factors $\beta_n = \nicefrac{1}{n}$ does not converge to the least
  fixpoint $\mu f = 0$, for any choice of the learning rates $\lr_n$
  \cut{(which, in a relaxed Mann-Kleene scheme must satisfy
    $\lim_{n \to\infty} \lr_n = \alpha > 0$)} (see
  Remark~\ref{rm:weighted-convergence} for the details of the calculation).
  Instead, by Corollary~\ref{cor:normal}, for any $0 < \epsilon < 1$
  the sequence obtained with learning rates
  $\lr_n\coloneqq \nicefrac{1}{n^\varepsilon}$ converges to $\mu f$.
\end{example}
 
\begin{toappendix}
  \begin{remark}
    \label{rm:weighted-convergence}
      Consider identity function $f\colon [0,1] \to [0,1]$ and a sequence
  of approximating functions $f_n\colon [0,1] \to [0,1]$ defined by
  $f_n(x) = (1-\nicefrac{1}{n})x+\nicefrac{1}{n}$.  Consider a sequence $(x_n)$ generated by a
  Mann relaxed Mann-Kleene with dampening factors $\beta_n = \nicefrac{1}{n}$,
  starting from $x_0 = 1$. Then one can see that $(x_n)$
  does not converge to the least fixpoint $\mu f = 0$, for whatever the choice of the
  learning rates $\lr_n$ (which, in a relax Mann-Kleene scheme must satisfy d
  $\lim_{n \to\infty} \lr_n = \alpha > 0$).
  In fact, by unrolling the definition of $f_n$ one obtains
    \[x_{n+1} = \left(1 - \frac1n\right) \cdot \Big[ x_n \Big(1- \frac{\lr_n}{n}\Big) + \frac{\lr_n}{n}\Big]\]

    Then, one can easily see that $x_1 = 0$ and $x_n > 0$ for all
    $n \geq 2$. Now, assuming that $x_n$ converges to
    $0$ one can derive that the sequence is eventually increasing, reaching a contradiction.

    More in detail, assume that $x_n \to 0$ and let
    $0 < \epsilon < \nicefrac{\alpha}{4}$. Then, there is $n_0$ such that for all
    $n \geq n_0$, it holds
  \begin{center}
    $|\alpha - \lr_n| \leq \epsilon$ and $x_n \leq \displaystyle \frac{\alpha-3\epsilon}{4}$
  \end{center}

  Therefore, for all $n \geq n_0$ it holds
  \begin{align*}
    x_{n+1}
    & = \left(1 - \frac1n\right) \cdot \Big[ x_n \Big(1- \frac{\lr_n}{n}\Big) + \frac{\lr_n}{n}\Big]
    & \mbox{[since $|\lr_n- \alpha| \leq \epsilon$]}\\
    & \geq \left(1 - \frac1n\right) \cdot \Big[ x_n \Big(1- \frac{\alpha+\epsilon}{n}\Big) + \frac{\alpha-\epsilon}{n}\Big]\\
    & = \left(1 - \frac1n\right) \cdot \Big[ x_n - \frac{\alpha+\epsilon}{n} x_n + \frac{\alpha-\epsilon}{n}\Big]
    & \mbox{[since $x_n \leq \frac{\alpha-3\epsilon}{4} \leq \frac{1}{4}$]}\\
    & \geq \left(1 - \frac1n\right) \cdot \Big[ x_n - \frac{\alpha+\epsilon}{4n} + \frac{\alpha-\epsilon}{n}\Big]\\
    & = \left(1 - \frac1n\right) \cdot \Big[ x_n + \frac{3\alpha -5\epsilon}{4n} \Big]\\
    & = x_n - \frac{1}{n} x_n + \left(1 - \frac1n\right) \frac{3\alpha -5\epsilon}{4n} \\
    & \geq x_n 
  \end{align*}
  where the last passage is motivated by the fact that 
  \begin{align*}
    & - \frac{1}{n} x_n + \left(1 - \frac1n\right) \frac{3\alpha -5\epsilon}{4n}
    & \mbox{[since $x_n \leq \frac{\alpha-3\epsilon}{4}$]}\\
    & \geq - \frac{1}{n} \frac{\alpha-3\epsilon}{4} + \left(1 - \frac1n\right) \frac{3\alpha -5\epsilon}{4n}\\    
    & = ((n-1) (3\alpha-5\epsilon) - n(\alpha-3\epsilon))/4n^2\\
    & = (2n (\alpha-\epsilon) - 3 \alpha + 5\epsilon)/4n^2\\
    & \geq 0
  \end{align*}
  Namely, the sequence is eventually increasing and always greater than $0$, contradicting the fact that it converges to $0$.
\end{remark}
\end{toappendix}

\subsection{Chaotic Dampened-Mann Iteration}
\label{ss:gen-mann}

We now generalize Mann schemes
by allowing parameters $\lr_n,\beta_n$ to depend not only on the step,
but also on the dimensions.
This establishes the basis for performing chaotic iteration to approximate the least fixpoint.

The formal switch is, again, apparently minor: dampening and learning parameters become vectors instead of scalars.

\begin{definition}[Generalized Mann scheme]
  \label{de:dampened-mann-gen}
  A \emph{generalized (dampened) Mann scheme}
  $\mathcal{S}=((\lr_n),(\beta_n))$ is a pair of
  parameter sequences where the \emph{learning rates} $(\lr_n)$ and
  \emph{dampening factors} $(\df_n)$ are sequences of
  \emph{vectors} in $[0,1)^d$, such that
  conditions (\ref{eq:beta-condition-1})
  $\lim_{n \to \infty} \beta_n = 0$ and (\ref{eq:beta-condition-2})
  $\sum_{n \in \mathbb{N}_0} \beta_n = \infty$ of
  Definition~\ref{de:dampened-mann} hold pointwise.
\end{definition}

Notions introduced for Mann schemes naturally extend to generalized
ones. The sequence generated by a generalized scheme $\mathcal{S}$ is still defined by
(\ref{eq:iteration}), but interpreted pointwise.
\cut{Note that the restriction
to sets $X$ to $0$-boxes (Condition~\eqref{eq:domain-form}) is now required for the
iteration to remain well-defined.}
Similarly, the notion of \emph{exactness} naturally extends to
generalized Mann schemes.

One might think that in order to ensure convergence one could just
extend the notion of {\reg} schemes to the generalized case by interpreting
Condition~\eqref{ass:quotient} pointwise.
The following example shows that this does not work in general.

\begin{example}
  Consider
  $f\colon [0,1]^2 \to [0,1]^2$ with
$f(x, y) = (\max(x,y), \max(x,y))$ and the (vector-)sequences
$(\alpha_n), (\beta_n)$ with $\alpha_n = (1, 1)$ and
\[\beta_n(1) = \begin{cases} \nicefrac{1}{(n+2)} & \text{if $n$ is even} \\ \nicefrac{1}{(n+2)^2} & \text{if $n$ is odd} \end{cases} \qquad \beta_n(2) = \begin{cases} \nicefrac{1}{(n+2)^2} & \text{if $n$ is even} \\ \nicefrac{1}{(n+2)} & \text{if $n$ is odd} \end{cases}\]
Clearly, $\mu f = (0, 0)$, but it is easy to see that the sequence
$(x_n)$ generated by dampened Kleene Iteration with these parameters
from the starting point $(1, 1)$ will not converge to $(0,0)$, because
the positive number $\prod_{n=2}^\infty (1-\nicefrac{1}{n^2})$ is a
lower bound for at least one component of all $x_n$.
\end{example}

The crux is that we need sufficiently strong dampening, uniformly
across all components. Otherwise, thinking of applications to
state-based systems, it might happen that two states \emph{convince
  each other} of a too large value by passing it on between each other
and, at each step, one of them preserves the over-approximation while
the other one gets dampened.

For this reasons, we extend the notion of {\reg} schemes as follows:
\begin{definition}[{\Reg} generalized Mann scheme]
  \label{de:generalized-Mann-scheme}
  A generalized Mann scheme
  $\mathcal{S}=((\lr_n),(\beta_n))$ is
  \emph{\reg}
  if condition below holds pointwise
  \begin{equation}
    \lim_{n\to\infty}\nicefrac{\beta_n}{\lr_n} = 0
  \end{equation}
  and there exists a strictly increasing sequence $(m_k)$ such that,
   for all $k\in\N_0$ and $i\in\{1,\dots,d\}$, 
  $U_k(i)\coloneqq\{\ell\in\{m_k,\dots,m_{k+1}-1\}\mid \lr_\ell(i)>0\}\neq\emptyset$
  and, letting $\ell_k(i)\coloneqq\max U_k(i)$,
  we have
  \begin{equation}\label{MainTheoremExtraCondition}
    \sum_k \min\{\beta_{\ell_k(i)}(i) \mid i=1,\dots,d\} = \infty
  \end{equation}
\end{definition}

Intuitively, $U_k(i)$ is the set of steps in
$\{m_k, \ldots, m_{k+1}-1\}$ in which component $i$ has been updated
and $\ell_k(i)$ is the last update in the interval. Hence
Condition~\eqref{MainTheoremExtraCondition} requires that the
dampening factors in each component decrease in some sense fairly:
There is a sequence of time steps such that in between any two
subsequent time steps, each component is updated at least once and,
considering only the last updates in each component, the worst
dampening factors still sum up to infinity.
This is a natural extension of the non-generalized case: a (non-generalized) progressing scheme using the same parameter sequences in each component
fulfills~\eqref{MainTheoremExtraCondition} for the sequence $m_k=k$.

\begin{toappendix}
The following result will allow us to restrict some of our proofs
to the case $X=\Rp^d$.

\begin{lemmarep}
  \label{lem:subset}
    Given a $0$-box $X\subseteq\Rp^d$ and a monotone and non-expansive function
    $f\colon X\to X$, then the function $F\colon \Rp^d\to\Rp^d$ given by
    \[F(x)=\sup\{f(y)\mid y\in X,y\leq x\}\]
    is a monotone and non-expansive extension of $f$ (with $\mu f=\mu\overline{f}$).
\end{lemmarep}
\begin{proof}
  First, we notice that $F$ is indeed an extension of $f$ since $f$ is monotone and thus,
  for $x\in X$, we have $f(x)\geq f(y)$ for all $y\leq x$, whence $F(x)=f(x)$.

  Furthermore, it is easy to see that $F$ is monotone since $x\leq x'$ implies
  $\{f(y)\mid y\in X,y\leq x\}\subseteq\{f(y)\mid y\in X,y\leq x'\}$.

  To show that $F$ is non-expansive, let $x'\in\Rpi^d$ be such that $X=\{x\in\Rp^d\mid x\leq x'\}$.
  We define the function $P_f:\Rp^d\to X$ as
  \[P_f(x)=\min(x,x')\]
  and note that $P_f$ is non-expansive and $F(x)=f(P_f(x))$.
  But then, for $x,\hat{x}\in\Rp^d$
  \[\|F(x)-F(\hat{x})\|\leq\|f(P_f(x))-f(P_f(\hat{x}))\|\leq\|P_f(x)-P_f(\hat{x})\|\leq\|x-\hat{x}\|.\]
  Lastly, since $F$ is an extension of $f$, we have $\mu F\leq\mu f$.
  But then, in particular, $\mu F\in X$ by definition of $X$ and thus we must have $\mu F=\mu f$.
  \qed
\end{proof}
\end{toappendix}

Our first main result is the convergence of generalized dampened Mann
iteration to the least fixpoint of $f$ when the exact function $f$ is
known.

\begin{maintheoremrep}[Exact Case]
  \label{thm:exact}
    {\Reg} generalized Mann schemes are exact.
\end{maintheoremrep}
\begin{proofsketch}
The statement is shown in two parts:
\begin{itemize}
\item First we show that $\limsup_{n\to\infty} x_n \le \mu f$. This is done by comparing $(x_n)$ to a sequence $(y_n)$ that is generated with the same parameters but starts from the bottom element $y_n = 0$. It is clear that $(y_n)$ always stays below the least fixpoint of $f$ and one can show that the distance $\|x_n - y_n\|$ between the two sequences converges to zero. Here we need our additional assumption that the dampening is uniform over all components.
\item In the more interesting part of the proof we show that $\limsup_{n\to\infty} x_n \ge \mu f$. To this end we construct, by transfinite induction, a sequence $(\boldsymbol{t}_\alpha)$ of \emph{stepping stones} below $\mu f$ with the property that $(x_n)$ is eventually above $\boldsymbol{t}_\alpha$. One might hope that $(\boldsymbol{t}_\alpha)$ is monotonically increasing and converges to $\mu f$. We do not achieve monotonicity in the proof but we show that the sequence can be constructed in a way that the sequence of sums $\boldsymbol{t}_\alpha(1) + \dots + \boldsymbol{t}_\alpha(d)$ is strictly increasing. That means, even if we loose progress towards $\mu f$ in one dimension, the overall progress in the other components still makes up for it. For this we need the property that the dampening factors decrease more rapidly than the learning-rates. The fact that there is always an overall gain in the sum of the components of $\boldsymbol{t}_\alpha$ (together with the first part of the proof) allows us to conclude that the sequence $(x_n)$ converges to $\mu f$.
  \qed
\end{itemize}
\end{proofsketch}
\begin{proof}
    By~\Cref{lem:subset}, there exists a monotone and non-expansive extension $\overline{f}$ of $f$ with
    the same least fixpoint.
    Since $x_n\in X$ for all $n\in\N_0$, the sequence generated by~\eqref{eq:iteration} stays the
    same when replacing $f$ with $\overline{f}$, whence we can, w.l.o.g., assume that $X=\Rp^d$.

    We first show that $\limsup_{n\to\infty} x_n \le \mu f$. To this end we look at the sequence $(y_n)$ with $y_0 = (0,\dots,0)$ and the same recursive rule
\[y_{n+1}(j) = (1- \beta_n(j)) \cdot \big( (1-\lr_n(j)) y_n(j) + \lr_n(j) f(y_n)(j)\big)\]
It is clear that $y_n \le \mu f$ for all $n\in\N$ and therefore it suffices to show that $\|y_m - x_m\| \to 0$. To this end we show by induction on $N$ that
\[\|y_{m_N} - x_{m_N}\| \le \|y_0 - x_0\| \cdot \prod_{n<N} \left(1 - \min\{\beta_{\ell_n(i)}(i) \mid i=1,\dots,d\}\right)\]
Assume that it holds for $N$ and consider $j$ with $1 \le j \le d$. We claim that
\[|y_{m_{N+1}}(j) - x_{m_{N+1}}(j)| \le \|y_0 - x_0\| \cdot \prod_{n<N+1} \left(1 - \min\{\beta_{\ell_n(i)}(i) \mid i=1,\dots,d\}\right)\]
Note that by non-expansiveness of $f$ and definition of the iteration scheme, we have
$\|y_{k+1}-x_{k+1}\|\leq\|y_k-x_k\|$ for any $k\in\N_0$.
In particular, by the inductive hypothesis, we have
\begin{align*}
  \|y_{\ell_N(j)} - x_{\ell_N(j)}\| & \le \|y_{m_N}-x_{m_N}\|\\
  & \leq \|y_0 - x_0\| \cdot \prod_{n<N} \left(1 - \min\{\beta_{\ell_n(i)}(i) \mid i=1,\dots,d\}\right)
\end{align*}
Consequently
\begin{align*}
  & |y_{\ell_N(j)+1}(j) - x_{\ell_N(j)+1}(j)| \\
  &\le \|y_0 - x_0\| \cdot (1-\beta_{\ell_N(j)}(j)) \prod_{n<N} \left(1 - \min\{\beta_{\ell_n(i)}(i) \mid i=1,\dots,d\}\right)\\
&\le \|y_0 - x_0\| \cdot \prod_{n<N+1} \left(1 - \min\{\beta_{\ell_n(i)}(i) \mid i=1,\dots,d\}\right)
\end{align*}
Since $\lr_q(j) = 0$ for all $q$ with $\ell_N(j) < q < m_{N+1}$ this implies that also
\[|y_{m_{N+1}}(j) - x_{m_{N+1}}(j)| \le \|y_0 - x_0\| \cdot \prod_{n<N+1} \left(1 - \min\{\beta_{\ell_n(i)}(i) \mid i=1,\dots,d\}\right)\]
Because $j$ was arbitrary, we can indeed conclude that
\[\|y_{m_{N+1}} - x_{m_{N+1}}\| \le \|y_0 - x_0\| \cdot \prod_{n<N+1} \left(1 - \min\{\beta_{\ell_n(i)}(i) \mid i=1,\dots,d\}\right)\]
    Since $\sum_n\min\{\beta_{\ell_n(i)}(i)\mid i=1,\dots,d\}=\infty$, it follows that $\|y_m-x_m\|\to0$ and thus $\limsup_{n\to\infty}x_n\leq\mu f$ as desired.

    Next, we show that $\liminf_{n\to\infty} x_n \ge \mu f$ if $\mu f<\infty$.
    We may assume that $\mu f \neq 0$ (otherwise the statement is clear).
    Furthermore, we note that there can only be finitely many $n$ with $\lr_n=0<\df_n$ whence we
    can, w.l.o.g., assume that $\df_n=0$ if $\lr_n=0$.

    By the above, the sequence $(x_n)$ is bounded, say $x_n \in [0, B]^d$ for all $n\in\N$ with some
    $B$ with $\mu f \le B$.
    We are going to define a sequence of $d$-tuples $\boldsymbol{t}_\alpha \in [0, B]^d$ with
    $\boldsymbol{t}_\alpha \le \mu f$ by transfinite recursion on ordinals $\alpha$ such that the
    following properties hold:
    \begin{enumerate}
        \item[(a)] The point $\boldsymbol{t}_\alpha$ is a post-fixpoint of $f$ (but not a fixpoint).
        \item[(b)] For all $\alpha$ there is $N\in\N$ such that $x_n \ge \boldsymbol{t}_\alpha$ for all $n \ge N$.
        \item[(c)] For all ordinals $\alpha < \beta$ we have $\boldsymbol{t}_\alpha(1) + \dots + \boldsymbol{t}_\alpha(d) < \boldsymbol{t}_\beta(1) + \dots + \boldsymbol{t}_\beta(d)$.
    \end{enumerate}

    Let $\boldsymbol{t}_0 = (0, \dots, 0)$, then (a) and (b) are trivially satisfied.

    For the successor step, assume that $\boldsymbol{t}_\alpha$ has been defined and choose $N\in\N$
    according to (b), i.e., with $x_n\ge\boldsymbol{t}_\alpha$ for all $n\ge N$.
    Let $\boldsymbol{s}\coloneqq f(\boldsymbol{t}_\alpha)$, then by (a) we have
    $\boldsymbol{t}_\alpha \le \boldsymbol{s}$ and the inequality is strict in at least one component.
    Pick $j$ with $\delta\coloneqq \boldsymbol{s}(j) - \boldsymbol{t}_\lr(j) > 0$ and put
    $\delta_n\coloneqq \lr_n(j) \delta$.
    Note that, by assumption~\eqref{ass:quotient}, we have
    \[\lim_{n\to\infty}\frac{\df_n(j)}{\lr_n(j)}\left( 2(\boldsymbol{t}_\alpha(j) + B) \frac{1}{\delta} + \lr_n(j) \right) = 0\]
    In particular, we can choose $M\in\N$ such that for all $n\ge M$ we have
    \begin{equation}\label{eq:proof-exact-1}
        \frac{\df_n(j)}{\lr_n(j)}\left( 2(\boldsymbol{t}_\alpha(j) + B) \frac{\lr_n(j)}{\delta_n} + \lr_n(j) \right) \le \frac12
    \end{equation}
    Since $\lr_n(j)$ is non-zero infinitely often, we can choose $K \ge \max\{N, M\}$ with $\lr_K(j) \neq 0$ and define $\boldsymbol{t}_{\alpha+1}$ by
    \[\boldsymbol{t}_{\alpha+1}\coloneqq (\boldsymbol{t}_\alpha(1), \dots, \boldsymbol{t}_\alpha(j) + \nicefrac{1}{2} \cdot \delta_K, \dots, \boldsymbol{t}_\alpha(d))\]

    Now, (c) is clearly satisfied.
    Also $\boldsymbol{t}_{\alpha+1}$ is a post-fixpoint because $\boldsymbol{t}_\alpha \le \boldsymbol{t}_{\alpha+1} \le \boldsymbol{s} = f(\boldsymbol{t}_\alpha)$ and therefore, by monotonicity of $f$, we have
    \[f(\boldsymbol{t}_{\alpha+1}) \ge f(\boldsymbol{t}_\alpha) = \boldsymbol{s} \ge \boldsymbol{t}_{\alpha+1}\]
    Finally, we show that (b) is satisfied for all $n \ge K+1$ by induction on $n$.
    Note that this is clear for all components but the $j$-th one, by the inductive hypothesis on
    $\alpha$, so that we can focus on the $j$-th component.
    Since, by the inductive hypothesis, $x_K \ge \boldsymbol{t}_\alpha$ and thus also
    $f(x_K)\ge\boldsymbol{s}$ by monotonicity of $f$, using $\lr_K(j)\in[0,1]$ as well
    as~\eqref{eq:proof-exact-1}, we get
    \begin{align*}
        x_{K+1}(j) &\ge (1- \df_K(j)) \cdot \big( \boldsymbol{t}_\alpha(j) + \lr_K(j)(\boldsymbol{s}(j)-\boldsymbol{t}_\alpha(j)) \big)\\
        &= (1- \df_K(j)) \cdot \big( \boldsymbol{t}_\alpha(j) + \delta_K \big)\\
        &= \boldsymbol{t}_\alpha(j) + \delta_K - \df_K(j)(\boldsymbol{t}_\alpha(j) + \delta_K)\\
        &= \boldsymbol{t}_\alpha(j) + \delta_K - \delta_K\frac{\df_K(j)}{\lr_K(j)}\left( \boldsymbol{t}_\alpha(j)\frac{\lr_K(j)}{\delta_K} + \lr_K(j) \right)\\
        &\ge \boldsymbol{t}_\alpha(j) + \frac12 \delta_K
    \end{align*}
    Now, assume that $x_n \ge \boldsymbol{t}_{\alpha+1}$.
    If $\lr_n(j)=0$, then by assumption $\df_n(j)=0$ as well and, therefore,
    $x_{n+1}(j) = x_n(j) \ge \boldsymbol{t}_\alpha(j) + \delta_K/2$.
    If $\lr_n(j)>0$, we put $e\coloneqq \delta - \delta_K/2 \ge \delta/2 > 0$ and
    $e_n\coloneqq \lr_n(j) e$.
    Then
    \begin{align*}
        x_{n+1}(j) &\ge (1- \df_n(j)) \cdot \big((\boldsymbol{t}_\alpha(j) + \delta_K/2) + \lr_n(j) (\boldsymbol{s}(j)-(\boldsymbol{t}_\alpha(j) + \delta_K/2)) \big)\\
        &= (1- \df_n(j)) \cdot \big( (\boldsymbol{t}_\alpha(j) + \delta_K/2) + e_n \big)\\
        &= (\boldsymbol{t}_\alpha(j) + \delta_K/2) + e_n - \df_n(j)(\boldsymbol{t}_\alpha(j) + \delta_K/2 + e_n)\\
        &= (\boldsymbol{t}_\alpha(j) + \delta_K/2) + e_n - e_n\frac{\df_n(j)}{\lr_n(j)} \left( \boldsymbol{t}_\alpha(j)\frac{\lr_n(j)}{e_n} + \delta_K/2 \frac{\lr_n(j)}{e_n} + \lr_n(j) \right)\\
        &\ge (\boldsymbol{t}_\alpha(j) + \delta_K/2) + e_n - e_n\frac{\df_n(j)}{\lr_n(j)} \left( (\boldsymbol{t}_\alpha(j)+B)\frac{\lr_n(j)}{e_n} + \lr_n(j) \right)\\
        &\ge (\boldsymbol{t}_\alpha(j) + \delta_K/2) + e_n - e_n\frac{\df_n(j)}{\lr_n(j)} \left( 2(\boldsymbol{t}_\alpha(j)+B)\frac{\lr_n(j)}{\delta_n} + \lr_n(j) \right)\\
        &\ge (\boldsymbol{t}_\alpha(j) + \delta_K/2) + e_n - \frac12 e_n\\
        &\ge (\boldsymbol{t}_\alpha(j) + \delta_K/2)
    \end{align*}
    where we used that $\delta_K \le \delta \le \boldsymbol{s}(j) \le \mu f(j) \le B$ and $2e_n\geq\delta\geq\delta_n$.

    For the limit step, let $\lambda$ be a (countable) limit ordinal and assume that
    $\boldsymbol{t}_\alpha$ has been defined for all $\alpha < \lambda$.
    Define $\boldsymbol{s}\coloneqq \sup_{\lr<\lambda} \boldsymbol{t}_\lr$, then by the inductive
    assumption (a), $\boldsymbol{s}$ is a post-fixpoint of $f$ (as a supremum of post-fixed
    points).
    We distinguish two cases:
    \begin{itemize}
        \item If $\boldsymbol{s} = \mu f$, then we can stop the recursion and finish the proof of
            $\liminf x_n \ge \mu f$ as follows: Given $\varepsilon>0$ we can choose for each index
            $i=1,\dots,d$ an ordinal $\lr_i < \lambda$ such that
            $\boldsymbol{t}_{\lr_i}(i) \ge \mu f(i) - \varepsilon$.
            Choose $N_i$ for $\lr_i$ as in (b), then for all $n\ge\max\{N_1, \dots, N_d\}$ we
            have $x_n \ge \mu f - \varepsilon$.
        \item If $\boldsymbol{s}<\mu f$, put $\boldsymbol{u}\coloneqq f(\boldsymbol{s})$.
            Then, $\boldsymbol{s} \le \boldsymbol{u}$ and the inequality is strict in at least one
            component, say $\delta\coloneqq \boldsymbol{u}(j) - \boldsymbol{s}(j) > 0$.
    \begin{enumerate}
    \item[(i)] Putting $r \coloneqq \boldsymbol{s}(j) + \delta/2$, $e \coloneqq \delta/4$ and $e_n \coloneqq \lr_n(j) e$, choose $M\in\N$ large enough such that for all $n\ge M$ we have
    \[\frac{\df_n(j)}{\lr_n(j)} \left( r\frac{\lr_n(j)}{e_n} + \lr_n(j) \right) \le \frac12\]
    \item[(ii)] Choose for $i=1,\dots,d$ an ordinal $\lr_i < \lambda$ such that $\boldsymbol{t}_{\lr_i}(i) \ge \boldsymbol{s}(i) - \delta/(4d)$.
    \item[(iii)] Choose $N_i\in\N$ large enough, such that $x_n \ge \boldsymbol{t}_{\lr_i}$ for all $n\ge N_i$.
    \end{enumerate}
Put $K \coloneqq \max\{N_1,\dots,N_d, M\}$ and define a sequence $(y_n)_{n \ge K}$ by $y_K = (B,\dots,B)$ and
\[y_{n+1}(j) = (1- \df_n(j)) \cdot \big( (1-\lr_n(j)) y_n(j) + \lr_n(j) f(y_n)(j) \big)\]
Then trivially $y_K \ge x_K$ and since the sequences follow the same recursion rule we have $y_n \ge x_n$ for all $n\ge K$. In particular $y_n \ge \boldsymbol{t}_{\lr_i}$ for all $i = 1,\dots,d$ and all $n\ge K$. We claim that
\[y_n \ge (\max_{i\in\{1,\dots,d\}} \boldsymbol{t}_{\lr_i}(1), \dots, \boldsymbol{s}(j) + \delta/2, \dots, \max_{i\in\{1,\dots,d\}} \boldsymbol{t}_{\lr_i}(d))\]
for all $n\ge K$. For all but the $j$-th component this is clear. We can therefore focus on the $j$-th component, the case $n=K$ being trivial, because $\boldsymbol{s}(j) + \delta/2 \le \boldsymbol{u}(j) \le B$. Assume that the claim holds for some $n$, then by the non-expansiveness of $f$ we have
\begin{align*}
y_{n+1}(j) &= (1- \df_n(j)) \cdot \big( (1-\lr_n(j)) y_n(j) + \lr_n(j) f(y_n)(j) \big)\\
&\ge (1- \df_n(j)) \cdot \big( (1-\lr_n(j)) y_n(j) + \lr_n(j) (f(\boldsymbol{s})(j) - \delta/4) \big)\\
\end{align*}
because, by the choice of the $\lr_i$, the inductive assumption implies in particular that
\[y_n \ge \boldsymbol{s} - \delta/4\]
We may assume that $\lr_n \neq 0$, for otherwise $y_{n+1}(j) = y_n(j)$. This yields
\begin{align*}
y_{n+1}(j) &\ge (1- \df_n(j)) \cdot \big( (1-\lr_n(j)) r + \lr_n(j) (r+e) \big)\\
&= (1-\lr_n(j)) r + \lr_n(j) (r+e) - \df_n(j) (1-\lr_n(j)) r + \lr_n(j) (r+e)\\
&= r + e_n - \df_n(j)(r + e_n)\\
&= r + e_n - e_n\frac{\df_n(j)}{\lr_n(j)} \left( r\frac{\lr_n(j)}{e_n} + \lr_n(j) \right)\\
&\ge  r
\end{align*}
As before we get
\[\lim_{n\to\infty} \|x_{K+n} - y_{K+n}\| = 0\]
In particular there is $L \ge K$ such that $\|x_{n} - y_{n}\| \le \delta/12$ for all $n\ge L$, then also
\[x_n(j) \ge y_n(j) - \delta/12 \ge \boldsymbol{s}(j) + \nicefrac{5}{12} \cdot \delta\]
Now define
\[\boldsymbol{t}_\lambda = (\max_{i\in\{1,\dots,d\}} \boldsymbol{t}_{\lr_i}(1), \dots, \boldsymbol{s}(j) + \nicefrac{5}{12} \cdot \delta, \dots, \max_{i\in\{1,\dots,d\}} \boldsymbol{t}_{\lr_i}(d))\]
Then (b) and (c) are clearly satisfied for $\lambda$. It remains to show that $\boldsymbol{t}_\lambda$ is a post-fixpoint of $f$. We check for this each component separately. For $i = 1,\dots,d j$ note that by inductive hypothesis
\[f(\boldsymbol{t}_\lambda) \ge f(\boldsymbol{t}_{\lr_i}) \ge \boldsymbol{t}_{\lr_i}\]
and therefore
\[f(\boldsymbol{t}_\lambda) \ge \sup_{i\in\{1,\dots,d\}} \boldsymbol{t}_{\lr_i}\]
Also by non-expansiveness of $f$ we have
\[f(\boldsymbol{t}_\lambda)(j) \ge f(\boldsymbol{s})(j) - \delta/4 = \boldsymbol{u}(j) - \delta/4 = \boldsymbol{s}(j) + \nicefrac{3}{4} \cdot \delta \ge \boldsymbol{s}(j) + \nicefrac{5}{12} \cdot \delta = \boldsymbol{t}_\lambda(j)\]
\end{itemize}
This finishes the inductive construction of the sequence $(\boldsymbol{t}_\lr)$. Finally, note that the induction must stop at some countable limit ordinal by landing in the first case of the above case distinction (and thus $\liminf_{n\to\infty} x_n \ge \mu f$). Otherwise, $\boldsymbol{t}_\lr$ would be defined for all $\lr<\omega_1$ where $\omega_1$ is the first uncountable ordinal. But then (c) implies that the map $\lr \mapsto \boldsymbol{t}_\lr(1) + \dots + \boldsymbol{t}_\lr(d)$ is an order preserving embedding of $\omega_1$ into $\R$, a contradiction.

    Finally, we prove that $\liminf_{n\to\infty}x_n\geq\mu f$ for general $\mu f\in\Rpi^d$.
    Let $q<1$ be arbitrary and consider the function $f_q=q\cdot f$.
    As $f$ is non-expansive, $f_q$ is a $q$-contraction.
    In particular, $\mu f_q\in\Rp^d$.

    Writing $(x_n^{(q)})$ for the sequence generated by~\eqref{eq:iteration} with $f_q$, by the
    above we thus have $\lim_{n\to\infty}x_n^{(q)}=\mu f_q$.
    Furthermore, since $f_q\leq f$, we also have $\liminf_{n\to\infty}x_n\geq\mu f_q$.

    Since $q<1$ was arbitrary and $f=\sup_{q<1}f_q$, the Scott-continuity of the least fixpoint
    operator implies $\liminf_{n\to\infty}x_n\geq\sup_{q<1}\mu f_q=\mu f$, proving the statement.
    \qed
\end{proof}

Dampened Mann iteration admits chaotic iteration as a special case.

\begin{definition}[Chaotic iteration]
  \label{de:chaotic-mann-iteration}
  Let $\mathcal{S}=((\lr_n),(\df_n))$ be a generalized Mann scheme and
  let $(f_n)$ be a sequence of monotone and non-expansive maps over a $0$-box $X\subseteq\Rp^d$.  A
  \emph{chaotic iteration} on $\mathcal{S}$ based on a sequence of index
  sets $(I_n)$, $I_n\subseteq\{1,\dots,d\}$, is defined as below,
  with $x_0\in X$ arbitrary, and
  \begin{equation}
    \label{eq:chaotic-iteration}
    \begin{cases}
      x_{n+1}(i)=(1-\beta_n(i))(x_n(i)+\lr_n(i)(f_n(x_n)-x_n)(i))&\text{for }i\in I_n\\
      x_{n+1}(j)=x_n(j)&\text{for }j\notin I_n
    \end{cases}
  \end{equation}
\end{definition}

The results for generalized dampened Mann iterations immediately imply
analogous results for chaotic dampened Mann iterations.
Here, assuming $(\df_n)$ pointwise monotonically decreasing, the
somewhat technical condition in \Cref{thm:exact} translates to the
intuitive requirement that the time interval between two complete
sweeps should not become arbitrary large.

\begin{corollaryrep}[Chaotic iteration]
  \label{cor:chaotic-mann-iteration}
  Let $\mathcal{S}=((\lr_n),(\beta_n))$ be a
  generalized Mann scheme with $(\df_n)$ pointwise monotonically decreasing.
  Let $(f_n)$ be a sequence of monotone and non-expansive maps on a
  $0$-box $X\subseteq\Rp^d$, converging
  to  $f$. Let  $(I_n)$ be a sequence of set of
  indices and inductively define $(m_k)$ by $m_0 = 0$ and
  $m_{k+1}$ the least $m\in\N$ such that $\bigcup_{n=m_k}^{m-1} I_n = \{ 1, \ldots, d\}$.
  If
  \begin{equation}
    \label{eq:sum-chaotic}
    \sum_{k\in\N} \min_i \beta_{m_k}(i) = \infty
  \end{equation}
  then the sequences $(x_n)$ generated by chaotic iteration converge
  to $\mu f$.
\end{corollaryrep}

\begin{proof}
  Note that chaotic Mann iteration can be seen as a
  generalized Mann iteration as follows. Given the generalized Mann scheme
  $\mathcal{S}=((\lr_n),(\beta_n))$ and the
  index sets $(I_n)$, the chaotic iteration $(x_n)$ generated
  by~\eqref{eq:chaotic-iteration} is exactly the generalized dampened
  Mann iteration~\eqref{eq:iteration} with parameters
  $((\alpha_n'), (\beta_n'))$ defined as:
  \[
    \lr_n'(j)=\begin{cases}
      \lr_n(j) &\text{if }j\in I_n\\
      0&\text{if }j\notin I_n
    \end{cases}
  \]
  and
  \[
    \df_n'(j)=\begin{cases}
      \df_n(j)&\text{if }j\in I_n\\
      0&\text{if }j\notin I_n.
    \end{cases}
  \]

  Using the notation of Definition~\ref{de:generalized-Mann-scheme},
  by construction, recalling that the $\lr_n(j)$ are all positive, we
  have that the sequence $m_k$ satisfies the condition
  \[
    \{\ell\in\{m_k,\dots,m_{k+1}-1\}\mid \lr_\ell'(i)>0\}\neq\emptyset
  \]

  Moreover, since $\ell^i_k < m_{k+1}$ and $\df_n$ is pointwise monotonically decreasing, for all fixed $i$ and $k$,
  \[
    \min_j \df_{m_{k+1}}(j) \leq \df_{m_{k+1}}(i) \leq \df_{\ell_k^i}'(i)
  \]
  and thus
  \[
    \sum_k \min\{\df_{\ell_k^i}'(i) \mid i=1,\dots,d\} \geq \sum_k \min_i \df_{m_{k+1}}(i) =  \infty
  \]

  Thus we conclude by \cref{thm:exact}.
  \qed
\end{proof}

Given a (non-generalized) Mann scheme,
a simple way of applying the scheme in a chaotic fashion consists in
associating to each component a local copy of the parameters and then
updating at each step a single (independently) randomly chosen state,
letting the parameters progress to the next value only for this state. It can be proved that, in
this way, starting from a progressing scheme, the generated sequence converges almost surely to the fixpoint (see \cref{co:random} in the appendix).
This  will be used in the numerical experiments in \S\ref{se:numerical}.

\begin{toappendix}
  \paragraph{Chaotic random iteration.}
  A simple way of applying a (non-generalized) Mann scheme
  $\mathcal{S}=((\lr_n),(\df_n))$ in a chaotic fashion consists in
  associating to each component a local copy of the parameters,
  selecting at each step, in a random uniform way a component to be
  updated, and letting the parameter progress to the next value only
  for the component which has been updated.

  More precisely, one maintains local counters $c(1), \dots, c(d)$ to each component, initialized to $0$ and at each step $n \geq 0$:

\begin{enumerate}
    \item Select  $i \in [1,d]$ uniformly at random
    \item Define
      \begin{equation}
        \label{eq:random-iteration}
        \begin{cases}      
          x_{n+1}(i)=(1-\beta_{c(i)})(x_n(i)+\lr_{c(i)}(f_n(x_n)-x_n)(i))\\
          x_{n+1}(j)=x_n(j)&\text{for }j\neq i
        \end{cases}
      \end{equation}

    \item Increment $c(i)$
\end{enumerate}

Then, again as an easy corollary, one can show that using a progressing Mann scheme, with monotonically decreasing $\df_n$, the iteration converges almost surely.

\begin{corollaryrep}[Chaotic random iteration]
  \label{co:random}
  Let $\mathcal{S}=((\lr_n),(\beta_n))$ be a
  progressing Mann scheme and let $(f_n)$ be a sequence of monotone
  and non-expansive maps over a $0$-box $X\subseteq\Rp^d$, converging to a function $f$. Then the
  sequence $(x_n)$ generated by (\ref{eq:random-iteration}) above
  converges to $\mu f$ almost surely.
\end{corollaryrep}

\begin{proof}
  Observe that each iteration realized as in
  \eqref{eq:random-iteration} can be seen as a chaotic Mann iteration.
  In fact, let us denote by $i_n$ the component $i$ to be updated
  chosen at iteration $n$ and consider the generalized Mann scheme
  $((\lr_n'),(\df_n'))$ defined by
  \begin{center}
    $\lr_n'(i) = \lr_{s(n,i)}$ and  $\df_n'(i) = \df_{s(n,i)}$
  \end{center}
  where $s(n,i) = | \{ n' \mid n' < n\ \land\ i_{n'} = i\}$ (number of
  times component $i$ has been selected up to iteration $n$).  Fixing
  as index sets $I_n=\{i_n\}$, the iteration in
  \eqref{eq:random-iteration} is the corresponding chaotic iteration.

  Also observe that, since $s(n,i) \leq n$ and $\df_n$ is decreasing, the parameters
  $\df_n'(i)$ are lower bounded by $\df_n$.

  Define the sequence $(m_k)_{k \geq 0}$ as in
  Corollary~\ref{cor:chaotic-mann-iteration}, i.e.,
  \begin{itemize}
  \item $m_0 = 0$
  \item $m_{k+1} = \min \big\{ N > m_k\mid \forall j \in [1,d].\ \exists n \in [m_k, N).\ i_n=j\big\}$
  \end{itemize}
  namely, $m_{k+1}$ is the least index $h$ such that all $d$
  components have been selected at least once in the interval
  $(m_k,h]$.

  We show that with probability $1$, we have
  \[
    \sum_{k=0}^\infty \df_{m_k} = \infty.
  \]
  and thus, recalling that $\df_{m_k} \leq \min_i \df_{m_k}'(i)$, we
  deduce that $\sum_{k\in\N} \min_i \beta_{m_k}(i) = \infty$ and we
  conclude by Corollary~\ref{cor:chaotic-mann-iteration}.

  Let $L_k = m_{k+1} - m_k$ for $k \geq 0$. Then the random variables
  $(L_k)_{k \geq 0}$ are independent and identically distributed with
  finite expected value $\mathbb{E}[L_i] = \ell$. In
  fact, requiring that between $m_k+1$ and $m_{k+1}$ all components
  have been updated at least once, is an instance of a coupon
  collector problem, hence  $\ell = d H_d$ where
  $H_d = \sum_{j=1}^d \frac{1}{j}$ is the $d$-th harmonic number.

  Note that one can write $m_k$ as $\sum_{i=0}^{k-1} L_i$ and, by the
  strong law of large numbers, given that the $L_i$ are independent
  and identically distributed, we have that almost surely
  \[
    \lim_{k \to \infty}\frac{\sum_{i=0}^{k-1} L_i}{k} = \ell
  \]
  Putting things together, almost surely
  \[
    \lim_{k \to \infty} \frac{m_k}{k} = \ell
  \]
  and thus there exists $k_0$ such that for all $k \geq k_0$:
  \[
    m_k \leq 2 \ell  k.
  \]

  Since $\df_n$ is decreasing, for $k \geq k_0$:
  \[
    \df_{m_k} \geq \df_{2 \ell k}.
  \]
   
  We show that the the subseries $\sum_{k=0}^\infty \df_{2 \ell k}$,
  where we write $2 \ell k$ for the integer part
  $\lfloor 2 \ell k \rfloor$, is divergent and thus we conclude since 
   \[
     \sum_{k=0}^\infty \df_{m_k} \geq \sum_{k=0}^\infty \df_{2 \ell k} = \infty
  \]
  
  In order to show that $\sum_{k=0}^\infty \df_{2 \ell k} = \infty$
  observe that for all $k$, since each interval
  $[2 \ell k, 2 \ell (k+1))$ contains at most $2\ell+1$ indices and
  $\beta_n$ is monotonically decreasing
  \[
    \df_{2 \ell k} \geq \frac{1}{2 \ell+1} \sum_{n= 2 \ell k}^{2 \ell (k+1)- 1}  \df_n
  \]
  and therefore, for all $m$
  \[
    \sum_{k=0}^m \df_{2 \ell k} \geq \frac{1}{2 \ell+1} \sum_{k=0}^m \sum_{n= 2 \ell k}^{2 \ell (k+1)- 1}  \df_n = \frac{1}{2 \ell+1} \sum_{n= 0}^{2 \ell (m+1)- 1}  \df_n
  \]
  which, since $\sum_{n=0}^\infty \df_n = \infty$, implies, as desired.
   \[
     \sum_{k=0}^\infty \df_{2 \ell k} = \infty
   \]
   \qed
\end{proof}

\end{toappendix}

\subsection{Working with Approximations}
\label{sec:approx}

\Cref{thm:exact} assumes that the monotone and non-expansive function $f$ of interest is known
exactly.
The benefits of the (generalized) dampened Mann iteration, however, lie in the case where $f$ can only be approximated by a sequence $(f_n)$ of monotone and non-expansive approximations.

The main result of this section is the following result
that gives an essential sufficient criterion for convergence.

\begin{maintheoremrep}[Approximated Case] %
  \label{thm:monotone}
  Let
  $\mathcal{S}=((\lr_n),(\beta_n))$
  be a {\reg} Mann scheme, 
  $X\subseteq\Rp^d$ a $0$-box,
  and let $(f_n)$ be a sequence of monotone and non-expansive functions $f_n\colon X\to X$, pointwise
  converging to
  $f\colon X\to X$.
  Then, the sequence $(x_n)$ generated by~$\mathcal{S}$ fulfills
  \begin{enumerate}
  \item \label{thm:monotone:1}
    $\liminf_{n\to\infty}x_n\geq\mu f$
  \item \label{thm:monotone:2}
    if $\lim_{n\to\infty}\mu(\sup_{k\geq n}f_k)=\mu f$, then $\lim_{n\to\infty}x_n=\mu f$
  \end{enumerate}
\end{maintheoremrep}

\begin{proof}
  \begin{enumerate}
  \item Consider the sequence $(g_n)$ of functions
    $g_n=\inf_{k\geq n}f_k$.  Then, for every $n\in\N_0$, the function
    $g_n$ is, by construction, monotone and non-expansive and fulfills
    $g_n\leq g_{n+1}$.  Furthermore, since $f_n\to f$ pointwise, also
    $g_n\to f$ pointwise from below whence, by Scott-continuity of the
    least fixpoint operator, also $\sup_{n\in\N_0}\mu g_n=\mu f$.
    
    Now, for $n\in\N_0$, the sequence $(y_m^{(n)})$ generated by~\eqref{eq:iteration}
    using $y_0=x_n$ as well as the function $g_n$ converges, by~\Cref{thm:exact}, to
    $\mu g_n$.
    However, we also have $y_m^{(n)}\leq x_{m+n}$ due to $g_n\leq f_k$ for all $k\geq n$.
    Therefore,
    \[\liminf_{n\to\infty}x_n\geq\sup_{n\in\N_0}\lim_{m\to\infty}y_m^{(n)}=\sup_{n\in\N_0}\mu g_n=\mu f,\]
            proving the statement.
        \item For $n\in\N_0$, write $g_n=\sup_{k\geq n}f_k$.
            By construction, all $g_n$ are monotone and non-expansive, and, by assumption, fulfill
            $\mu g_n\to\mu f$.

            Analogously to the argumentation above, we thus get $\limsup_{n\to\infty}x_n\leq\mu f$.
            But then, using the above, we already have $\lim_{n\to\infty}x_n=\mu f$.
            \qed
    \end{enumerate}
\end{proof}

This result, in particular, implies convergence to the least fixpoint
if $(f_n)$ converges monotonically and $\mu f_n\to\mu f$, thus
generalizing~\cite[Theorem 4.6.1]{BGKPW:CAFDM}. Instead, in general,
the fact $\mu f_n\to\mu f$ is not sufficient to have convergence, as
shown by the example below, inspired by~\cite[Example
4.4]{BGKPW:CAFDM}.

\begin{example}
  Consider the function $f(x,y)=(y,x)$ and the sequence of
  approximations $(f_n)$ given by
  \[
    f_n(x,y)=\begin{cases}
      (y,x)&\text{if }n\text{ even}\\
      (y\ominus \nicefrac{2}{n},x\oplus \nicefrac{2}{n})&\text{if }n\text{ uneven}
    \end{cases}
  \]
  Despite $\mu f_n$ converging to $\mu f = (0,0)$, it can be seen that
  dampened Mann iteration does not converge in general.  Indeed, the
  hypotheses of \Cref{thm:monotone}(\ref{thm:monotone:2}) are
  violated, since the supremum $F_n=\sup_{k\geq n}f_k$ is always of
  the form
  $F_n(x,y)=(y,x+\varepsilon)$
  for some $\varepsilon>0$, thus having
  fixpoint $\mu F_n=(\infty,\infty)\nrightarrow (0,0)$.
\end{example}

By relying on \Cref{thm:monotone} one can also prove convergence in the approximated case when the limit
function $f$ is sufficiently well-behaved (without  assumptions on the mode of convergence).

\begin{theoremrep}[Convergence for Picard limit functions]
  \label{Theorem:PicardLimit}
  Let $(f_n)$ be a sequence of monotone and non-expansive functions
  $f_n\colon X\to X$ with $X$ a $0$-box, converging pointwise to a Picard operator $f\colon X\to X$.
  Then, we have
  \[\lim_{n\to\infty}\mu(\sup_{k\geq n}f_k)=\mu f.\]
  As a consequence, for any {\reg} Mann scheme $\mathcal{S}$, the sequence $(x_n)$ generated by
  $\mathcal{S}$ on $(f_n)$ converges to $\mu f$.
\end{theoremrep}

\begin{toappendix}
For the proof of the theorem above we require the following lemma.
For two points $p$ and $q$ of $\Rp^d$ we write $p \prec q$ if in every dimension
$i\in\{1,\dots,d\}$ we have a strict inequality $p(i) < q(i)$.
\begin{lemmarep}
Let $f\colon X \to X$ be a monotone and non-expansive map over a $0$-box $X\subseteq\Rp^d$.
If $q\in X$ is a point with $\mu f \le q$ such that $f^n(q) \prec q$ for some $n\in\N$ of the form
$n = 2^k$, then there is $p$ with $\mu f \le p \le q$, such that $f(p) \prec p$.
\end{lemmarep}
\begin{proof}
Fix a point $q$ with $\mu f \prec q$ and a power of two, say $n = 2^k$, with $f^n(q) \prec q$. We prove the conclusion of the lemma by induction on $k$ (simultaneously for all monotone and non-expansive maps).

If $k=0$ then the point $q$ is already as desired. Otherwise the monotone and non-expansive map $f^2$ satisfies $\mu (f^2) = \mu f$ and $(f^2)^{n/2}(q) \prec q$. By the inductive hypothesis there is a point $p$ with $\mu f \le p \le q$ such that $f^2(p) \prec p$. If $f(p) \prec p$ we are done. Otherwise the sets
\begin{eqnarray*}
  L & \coloneqq & \{i\in \{1,\dots,d\} \mid f(p)(i) < p(i)\} \\
  R & \coloneqq & \{i\in \{1,\dots,d\} \mid f(p)(i) \ge p(i)\}
\end{eqnarray*}
are both non-empty. In this case, we consider the point $r \coloneqq \inf\{p, f(p)\}$ (hence, by monotonicity, $f(r) \leq f(p)$ and $f(r) \leq f^2(p)$) and we claim that $f(r)(i) \le r(i)$ for all $i\in\{1,\dots,d\}$ and that the inequality is strict for all $i\in R$.
\begin{enumerate}
\item If $i\in R$, then $f(r)(i) \le \min\{f(p)(i), f^2(p)(i)\} = f^2(p)(i) < p(i) = r(i)$.
\item If $i\in L$, then $f(r)(i) \le \min\{f(p)(i), f^2(p)(i)\} \le f(p)(i) = r(i)$.
\end{enumerate}
Now define
\begin{align*}
\varepsilon &\coloneqq \min\{r(i) - f(r)(i) \mid i\in R\} > 0\\
\delta &\coloneqq \min\{p(i) - f(p)(i) \mid i\in L\} > 0
\end{align*}
Let $\gamma \coloneqq \min\{\varepsilon/3, \delta\}$, we claim that the point $s$ with
\[s(i) = \begin{cases}
r(i) & \text{if } i\in R\\
r(i) + \gamma & \text{if } i\in L
\end{cases}\]
is as desired. Note that $s \le p$ and $s \le r + \gamma$.
\begin{enumerate}
\item If $i\in R$, then using the non-expansiveness of $f$ we have
\[f(s)(i) \le f(r)(i) + \gamma < r(i) - \varepsilon + \gamma < r(i) = s(i)\]
\item If $i\in L$, then using the monotonicity of $f$ we have
  \[f(s)(i) \le f(p)(i) < f(p)(i) + \gamma = r(i) + \gamma = s(i)\]
  \qed
\end{enumerate}
\end{proof}

\begin{corollaryrep}\label{cor:picard-point}
Let $f\colon X \to X$ be a monotone and non-expansive Picard-operator over a $0$-box
$X\subseteq\Rp^d$.

Then, for all $\varepsilon>0$ with $\mu f+\varepsilon\in X$, there is a point $p$ with
$\mu f \le p \le \mu f + \varepsilon$ such that $f(p) \prec p$.
\end{corollaryrep}
\begin{proof}
Let $\varepsilon > 0$ and define $q \coloneqq \mu f + \varepsilon$.
Since $f$ is a Picard-operator the sequence $(f^n(q))$ converges to $\mu f$, in particular there is
$n\in\N$ of the form $n=2^k$, such that $f^n(q) \prec q$.
By the above lemma there is a point $p$ with $\mu f \le p \le q$ and $f(p) \prec p$.
\qed
\end{proof}

We are now ready to prove \Cref{Theorem:PicardLimit}.

\begin{proof}[Proof of \Cref{Theorem:PicardLimit}]
We write $g_n=\sup_{k>n}f_k$ for $n\in\N_0$ and let $\varepsilon>0$ be arbitrary.

By~\Cref{cor:picard-point}, there is a point $p$ with $\mu f \le p \le \mu f + \varepsilon$, such
that $f(p) \prec p$.
Define $\delta\coloneqq \min \{p(i) - f(p)(i) \mid i\in\{1,\dots,d\}\}$, then $\delta>0$ and there is
$N>\N$, such that $\|f_n(p)-f(p)\| \le \delta/2$ for all $n>N$.
In particular, we have $f_n(p) \le f(p) + \delta/2 \prec p$ for all $n>N$.
It follows that also $g_n(p) = \sup_{k>n}f_n(p) \le p$ for all $n>N$.
This means that $p$ is a pre-fixpoint of $g_n$, in particular $\mu g_n \le p$ for all $n>N$.
Since $\varepsilon>0$ was arbitrary, this shows $\lim_{n\to\infty}\mu g_n=\mu f$ and, thus, the
desired statement.

By~\Cref{thm:monotone}, it follows that any {\reg} generalized Mann scheme applied to $(f_n)$
yields an iteration converging to the fixpoint of $f$.
\qed
\end{proof}
\end{toappendix}

Note that this result implies convergence, in particular, for cases
where the limit function is contractive or power contractive, thus
generalizing both~\cite[Theorem 4.3 and Theorem B.3]{BGKPW:CAFDM-arxiv}.

\section{Applications to State-based Systems}
\label{se:ssg}

We treat the convergence of dampened Mann iterations for piecewise
linear functions such as those arising when dealing with optimal returns and optimal policies in zero-sum two-player
(turn-based) stochastic games.
These results generalize those in~\cite{BGKPW:CAFDM} which are limited
to MDPs, with proofs relying on techniques specifically designed for this setting,
like collapsing end-components, making it difficult to generalize them to
other function classes. Furthermore, we lift the requirement that the
least fixpoint, i.e., the expected return is finite for all
states. Finally, in~\cite{BGKPW:CAFDM}, we showed convergence only for
Mann-Kleene schemes, while here we consider generalized
{\bf} Mann schemes which allow also for asynchronous iteration.
\cut{All of
this will be done without the need to work with or even define notions
like the collapsing or (in the case of SSGs) deflating of end
components.}

\subsection{Simple Stochastic Games}

We will be working with the following definition of simple stochastic
games.

\begin{definition}[Simple stochastic game]
  A zero-sum two-player (turn-based) stochastic game $\mathcal{G}$
  (short: simple stochastic game or SSG) is a tuple
  $(S,A,T,R,S_{\max},S_{\min})$ where
    \begin{itemize}
    \item $S$ is a (finite) set of states
    \item $(A(s))_{s\in S}$ is a family of (finite) sets of actions
      $A(s)$ enabled in state $s$, and, by abuse of notation, we write
      $A\coloneqq\bigcup_{s\in S}A(s)$;
    \item
      $T\colon S\times A\times S\to[0,1],(s,a,s')\mapsto T_a(s,s')$ is
      a transition function fulfilling;
      \[T_a(s)\coloneqq\sum_{s'\in S}T_a(s,s')\in[0,1]\quad\text{ for
          all }s\in S,a\in A\]
    \item $R\colon S\times A\to\mathbb{R}_+,(s,a)\mapsto R_a(s)$ is a
      reward function;
    \item $S_{\max}\cup S_{\min}=S$ are a partition of $S$ where
      $S_{\max}$ and $S_{\min}$ are states controlled by the $\max$-
      and $\min$-player, respectively.
    \end{itemize}
    We write $F=F(\mathcal{G})\coloneqq\{s\in S\mid A(s)=\emptyset\}$
    for the set of final states.
\end{definition}
Intuitively, an SSG is a probabilistic state-based system where, in a
given non-final state $s\in S\setminus F$, one of two external agents
(the $\max$-player if $s\in S_{\max}$ or $\min$-player if
$s\in S_{\min}$) chooses one of the actions $a\in A(s)$ enabled in
said state after which the system emits a reward $R_a(s)$ and
transitions to a new state $s'\in S$ with probability $T_a(s,s')$ (or
terminates with probability $1-T_a(s)$).  The goal of the two players
is to choose actions in order to maximize ($\max$-player) or minimize
($\min$-player) the total expected reward.

As we allow $T_a(s)<1$, this definition also encompasses systems with
a discount factor.  Furthermore, it comprises as special cases also
Markov chains (where $|A(s)|\leq 1$ for all $s\in S$), (maximizing)
Markov decision processes (MDPs; where $|A(s)|\leq 1$ for all
$s\in S_{\min}$), and minimizing Markov decision processes (Min-MDPs;
where $|A(s)|\leq 1$ for all $s\in S_{\max}$).

Policies, as defined below, resolve the non-determinism present in SSGs by fixing the agents' behaviour.
\begin{definition}[Policy]
  Let $(S,A,T,R,S_{\max},S_{\min})$ be a SSG.
  A (memoryless and deterministic) policy on a set
  $S'\subseteq S\setminus F$ is a map $\pi\colon S'\to A$ with
  $\pi(s)\in A(s)$ for $s\in S'$.
  Given an SSG $\mathcal{G}=(S,A,T,R,S_{\max},S_{\min})$ and a policy
  $\pi\colon S'\to A$, we write
  $\mathcal{G}^\pi=(S,A^\pi,T,R,S_{\max},S_{\min})$ for the SSG where,
  for $s\in S$
  \[
    A^\pi(s)=
    \begin{cases}
      \{\pi(s)\}&\text{if }s\in S'\\
      A(s)&\text{otherwise}
    \end{cases}
  \]
\end{definition}
Note that, when $\pi_{\min}\colon S_{\min}\setminus F\to A$ and $\pi_{\max}\colon S_{\max}\setminus F\to A$ are
policies of the $\min$- and $\max$-player, respectively, we have that $\mathcal{G}^{\pi_{\min}}$ is
a (Max-)MDP, $\mathcal{G}^{\pi_{\max}}$ is a Min-MDP, and $\mathcal{G}^{(\pi_{\min},\pi_{\max})}$ is a
Markov chain (where $(\pi_{\min},\pi_{\max})$ denotes the combined policy for both players).

The expected total reward for each starting state $s$ in
$\mathcal{G}$, assuming that both players act optimally, can be
characterized as the least fixpoint (in $\Rpi^S$) of the so-called
Bellman operator
$f_{\mathcal{G}}\colon \mathbb{R}_+^S\to\mathbb{R}_+^S$ given
by
\[f_{\mathcal{G}}(v)(s)=\begin{cases}
    \max_{a\in A(s)}(R_a(s)+\sum_{s'\in S}T_a(s,s')v(s'))&\text{if }s\in S_{\max}\\
    \min_{a\in A(s)}(R_a(s)+\sum_{s'\in S}T_a(s,s')v(s'))&\text{if }s\in S_{\min}
  \end{cases}\]
Note that the expected total reward might be infinite.
\begin{toappendix}
  \begin{remark}
    \label{rem:gen-mdp}
    
  It should also be noted that
  the definitions of MDPs and Bellman operators in~\cite{BGKPW:CAFDM}
  and these do not exactly match. In~\cite{BGKPW:CAFDM}, we used rewards for
  transitions $(s,a,s')$ while we use rewards for actions
  $(s,a)$. This does, however, not result in a loss of theoretical
  expressiveness \emph{(every Bellman operator in~\cite{BGKPW:CAFDM}
    is a Bellman operator in this sense)}.  Here, we, on the other hand,
  allow arbitrary sub-probability distributions over next states
  ($T_a(s)\leq 1$ interpreted as the system terminating immediately
  with probability $1-T_a(s)$) which cannot be (directly) expressed by
  the definition in~\cite{BGKPW:CAFDM} \emph{(not every Bellman
    operator in this sense is a Bellman operator
    in~\cite{BGKPW:CAFDM})}.
\end{remark}
\end{toappendix}

For finite-state SSGs, it is known that memoryless policies such as
the ones described above are sufficient to act optimally (see,
e.g.,~\cite{c:complexity-stochastic-game}) whence we have
\begin{equation}\label{eq:value-equiv}
    \mu f_{\mathcal{G}}=\mu f_{\mathcal{G}^{\pi_{\min}^*}}=\mu f_{\mathcal{G}^{\pi_{\max}^*}}=\mu f_{\mathcal{G}^{\pi_{\max}^*,\pi_{\min}^*}}
\end{equation}
for some (optimal) policies $\pi_{\min}\colon S_{\min}\setminus F\to A$ and
$\pi_{\max}\colon S_{\max}\setminus F\to A$.

Before we give convergence proofs, we end this subsection by giving a definition of a sampled SSG:

\begin{definition}[Sampling]
  \label{def:sampling}
  Given an SSG $\mathcal{G}=(S,A,T,R,S_{\max},S_{\min})$, a sequence
  $(\mathcal{G}_n)$ of SSGs
  $\mathcal{G}_n=(S,A,T_n,R_n,S_{\max},S_{\min})$ is a \emph{sampling}
  of $\mathcal{G}$ if
    \begin{itemize}
        \item $T_n\rightarrow T$ and $R_n\to R$ for $n\to\infty$.

        \item For $s,s'\in S,a\in A$, if $T_a(s,s')=0$, then
          $(T_n)_a(s,s')=0$ for all $n\in\mathbb{N}$, i.e., a
          non-existing transition can never be sampled.

        \item For $s\in S,a\in A$, if $R_a(s)=0$, then $(R_n)_a(s)=0$
          for all $n\in\mathbb{N}$, i.e., a non-existing reward can
          never be sampled.

        \item For $s\in S,a\in A$, if $T_a(s)=1$, then $(T_n)_a(s)=1$
          for all $n\in\mathbb{N}$, i.e., a non-terminating action can
          never be sampled as terminating.
    \end{itemize}
\end{definition}

\subsection{Convergence for Simple Stochastic Games}

The proof of convergence of the dampened Mann iteration for SSGs will
mainly be based on the technical convergence criterion given
by~\Cref{thm:monotone}.

To be precise, we will show that the sequences of Bellman operators $(f_n)$
of sampled SSGs fulfill the condition
\begin{equation}
  \label{eq:conv-condition}
    \lim_{n\to\infty}\mu(\sup_{k\geq n}f_k)=\mu(\lim_{n\to\infty}f_n).
\end{equation}

This will be done by first reducing the problem to MDPs which, in turn,
depends on results for Markov chains, using~\eqref{eq:value-equiv}.
In fact, we use the following result for the optimal values of Markov
chains.

\begin{toappendix}
  We start by proving \enquote*{sampling-continuity} for Markov
  chains.  We fix a Markov chain $M=(S,T,R)$ where
  $T\colon S\times S\to [0,1]$ and $R\colon S\to \mathbb{R}_+$.  Note
  that, for Markov chains, the Bellman operator is of the form
    \[f(v)(s)=\begin{cases}
        0&\text{if }s\in F\\
        R_*(s)+\sum_{s'\in S}T_*(s,s')v(s')&\text{if }s\notin F\\
    \end{cases}\]
    where $*$ is the unique action in $A(s)$.
    As such, we can write $f$ as an affine linear function
    \[f(v)=Tv+R\]
    for some sub-probability matrix $T\in[0,1]^{S\times S}$ and vector $R\in\mathbb{R}_+^S$, and fully
    describe the Markov process by the tuple $(S,T,R)$.

    We recall the following basic concepts from Markov chain theory.
    \begin{definition}
        Given a Markov chain $M=(S,T,R)$, we define
        \begin{itemize}
        \item the relation $\to$ on states $S$ as $s\to s'$ iff
          $T(s,s')>0$;
        \item a state $s\in S$ to be essential iff for all $s'\in S$ with
          $s\to^* s'$, we have $T(s')=1$ and $s'\to^* s$, where $\to^*$ is
          the transitive and reflexive closure of $\to$
        \end{itemize}
    \end{definition}

    For the sake of readability, we will write inductively for $s,s'\in S$:
    \[\mathbb{P}(s\to^n s')=\begin{cases}
        1&\text{if }n=0,s=s'\\
        0&\text{if }n=0,s\neq s'\\
        \sum_{s''\in S}T(s,s'')\mathbb{P}(s''\to^{n-1}s')&\text{if }n>0
      \end{cases}\] where
    $f^n(v)(s)=\sum_{s'\in S}\left[\mathbb{P}(s\to^n
      s')v(s')+\sum_{k=0}^{n-1}\mathbb{P}(s\to^k s')R(s')\right]$ and
    $\mu f(s)=\sum_{k=0}^\infty\sum_{s'\in S}\mathbb{P}(s\to^k
    s')R(s')$ which corresponds to the traditional definition of the
    value.

    It is a well-known property that for essential states $s\in S$
    \[\sum_{n=0}^\infty\mathbb{P}(s\to^n s)=\infty,\]
    i.e., if the Markov process is in an essential state, it is
    expected to visit said state infinitely often (see,
    e.g.,~\cite{mc}).  On the other hand, for all $s\in S$, there
    either exists a state $s'\in S$ with $s\to^* s'$ that is either
    essential or fulfills $T(s')<1$.

    We now show that the value of some states can be purely determined by the abstract
    \enquote*{structure} of the MDP.
    Namely, for a Markov chain $M=(S,T,R)$, we define the sets
    \begin{itemize}
        \item $S_0\coloneqq\{s\in S\mid\forall s'\in S : s\to^* s'\Rightarrow R(s')=0\}$,
        \item $S_\infty\coloneqq\{s\in S\mid\exists s'\in S\setminus
          S_0 : s\to^* s'\land s'\text{ essential}\}$,
        \item $S_?\coloneqq S\setminus(S_0\cup S_\infty)$.
    \end{itemize}
    Note that any two Markov chains $M=(S,T,R)$ and $M'=(S,T',R')$ on the same state set fulfilling
    for all $s,s'\in S$
    \begin{itemize}
        \item $T(s,s')=0$ iff $T'(s,s')=0$
        \item $T(s)=1$ iff $T'(s)=1$
        \item $R(s)=0$ iff $R'(s)=0$
    \end{itemize}
    have the same sets $S_0,S_\infty,$ and $S_?$.

    First, we show the following lemma, showing that the Bellman operator of Markov chains is
    \enquote*{almost} a power contraction.
    \begin{lemmarep}\label{lem:mp-power-cont}
        Given a Markov chain $M=(S,T,R)$, the function
        \[\tilde{f}_M\colon \mathbb{R}_+^S\to\mathbb{R}_+^S,\tilde{f}_M(v)(s)=\begin{cases}
            0&\text{if }s\in S_0\cup S_\infty\\
            f_M(v)(s)&\text{if }s\in S_?
        \end{cases}\]
        is a power contraction.
    \end{lemmarep}
    \begin{proof}
        Write $f=\tilde{f}_M$ for convenience.
        Let $v,v'\in\mathbb{R}_+^S$ be arbitrary and write
        $p=\min(\min\{T(s,s')\mid s,s'\in S,s\to s'\},1-\max\{T(s)\mid s\in S,T(s)<1\})$.
        We start by noting that, by definition, $S_?$ contains no essential states.
        As such, we have for $s\in S_?$ that there exists $n\in\mathbb{N}_0$ and $s'\in S$ with
        $s\to^n s'$ and either $s'\in S_0\cup S_\infty$ or $T(s')<1$.

        Assume that $n$ is minimal with regards to that property (obviously $n<|S|$).
        Then, we have
        \begin{eqnarray*}
            &&\|f^{n+1}(v)(s)-f^{n+1}(v')(s)\| \\
            &\leq&\sum_{t\in S}\mathbb{P}(s\to^n t)\|f(v)(t)-f(v')(t)\|\\
            &\leq&(1-\mathbb{P}(s\to^n s'))\|v-v'\|+\mathbb{P}(s\to^n s')\|f(v)(t)-f(v')(t)\|
        \end{eqnarray*}
        But, if $s'\in S_0\cup S_\infty$, we have $\|f(v(s'))-f(v')(s')\|=0$ whence
        \[\|f^n(v)(s)-f^n(v')(s)\|\leq (1-\mathbb{P}(s\to^n s'))\|v-v'\|\leq (1-p^n)\|v-v'\|.\]
        If, on the other hand, $T(s')<1$, we have $\|f(v(s'))-f(v')(s')\|\leq T(s')\|v-v'\|$ whence
        \[\|f^n(v)(s)-f^n(v')(s)\|\leq (1-\mathbb{P}(s\to^n s')(1-T(s')))\|v-v'\|\leq (1-p^{n+1})\|v-v'\|.\]
        All in all, as $s\in S_?$ was arbitrary and $\|f(v)(s)-f(v')(s)\|=0$ for any $s\in S_0\cup
        S_\infty$, we have
        \[\|f^{|S|}(v)-f^{|S|}(v')\|\leq(1-p^{|S|})\|v-v'\|,\]
        showing that $f$ is a power-contraction.
        \qed
    \end{proof}

    Now, we can show that the sets above indeed accurately describe the value:
    \begin{lemma}\label{lem:mp-value-struct}
        Given a Markov chain $M=(S,T,R)$, we have for $s\in S$ that
        \begin{enumerate}
            \item $\mu f_M(s)=0$ if and only if $s\in S_0$
            \item $\mu f_M(s)=\infty$ if and only if $s\in S_\infty$
        \end{enumerate}
    \end{lemma}
    \begin{proof}
        \begin{enumerate}
            \item Note that, for $s\in S_0$, we have $R(s)=0$ and $s'\in S_0$ for all $s'\in S$ with
                $s\to s'$.
                In particular, if $v\in[0,\infty]^S$ with $v(s)=0$ for all $s\in S_0$, we also have
                \[f_M(v)(s)=R(s)+\sum_{s'\in S}T(s,s')v(s')=\sum_{s'\in S_0}T(s,s')v(s')=0\]
                for all $s\in S_0$, proving $\mu f_M(s)=0$ for all $s\in S_0$.

                On the other hand, for $s\in S\setminus S_0$, there exists $n\in\mathbb{N}_0$ and
                $s'\in S$ with $s\to^n s'$ and $R(s')>0$.
                But then
                \[\mu f(s)\geq f^{n+1}(0)(s)=\sum_{t\in S}\sum_{k=0}^{n-1}\mathbb{P}(s\to^k t)R(t)\geq\mathbb{P}(s\to^n s')R(s')>0.\]
            \item Notice that by the above $S_\infty=\{s\in S\mid\exists s'\in S\colon s\to^* s'\land r(s')>0\land s'\text{ essential}\}$.

                Thus, let $s\in S_\infty$ and $s'\in S$ essential with $r(s')>0$ and $s\to^n s'$ for
                some $n\in\mathbb{N}_0$.
                But then, we immediately get
                \[\mu f_M(s)\geq\sum_{k=0}^\infty\mathbb{P}(s\to^k s')R(s')\geq\mathbb{P}(s\to^n s')\sum_{k=0}^\infty\mathbb{P}(s'\to^k s')R(s')=\infty.\]

                It remains to show that $\mu f_M(s)=\infty$ only if $s\in S_\infty$.
                For this, we notice that $S\setminus S_\infty$ forms a closed set in the sense
                \[\forall s\in S\setminus S_\infty,s'\in S\colon s\to^*s'\Rightarrow s'\notin S_\infty.\]
                Thus, we have for all $s\in S_?\cup S_0$ that $(f_M^{(0)})^n(0)(s)=f_M^n(0)(s)$ (with
                $f_M^{(0)}$ as in~\Cref{lem:mp-power-cont} above).
                But then, we have that $\mu f_M(s)=\mu f_M^{(0)}(s)$ for all $s\in S_?\cup S_0$.
                Since $f_M^{(0)}$ is a power-contraction by~\Cref{lem:mp-power-cont}, it has a (unique) fixpoint
                in $\mathbb{R}_+^S$, whence $\mu f_M(s)=\mu f_M^{(0)}(s)<\infty$ for
                $s\in S\setminus S_\infty$.
                \qed
        \end{enumerate}
    \end{proof}

    This shows that for a Markov process $M$, the fact that the value of $\mu f_M$ on some state is $0$ or $\infty$ is only
    determined by the underlying graph structure and by the support of its
    reward function. Therefore, for any sampling $(M_n)$ of a Markov chain $M$, if
    $\mu f(s) = 0/\infty$ then eventually $\mu f_n(s) = 0/\infty$.

    We can use this to show convergence of least fixpoints for sampled Markov chains.
\end{toappendix}

\begin{lemmarep}
  \label{lem:mp-sampled-lfp-convergence}
  Given a sampling $\mathcal{M}=(M_n)$ of a Markov chain $M$,
  we have
    $\mu f_M=\lim_{n\to\infty}\mu f_{M_n}.$
\end{lemmarep}
\begin{proof}
  For the sake of readability, we write $f_n=f_{M_n}$ and $f=f_M$.
  Note that by definition of sampling, we have for $n$ large enough
  $T_n(s,s')>0$ iff $T(s,s')>0$, $T_n(s)<1$ iff $T(s)<1$, and
  $R_n(s)>0$ iff $R(s)>0$ whence also $(S_n)_0=S_0$ and
  $(S_n)_\infty=S_\infty$.  W.l.o.g., we assume that this holds true
  for all $n\in\mathbb{N}$.
  
  Lemma~\ref{lem:mp-value-struct} thus gives us
  \[\mu f_n(s)=0\Leftrightarrow\mu f(s)=0\quad\text{ and }\quad\mu f_n(s)=\infty\Leftrightarrow\mu f(s)=\infty.\]
  It remains to be proven that $\mu f_n(s)\to\mu f(s)$ for $s\in S_?$.
  Now, as above we make use of the fact that, for
  $s\in S\setminus S_\infty$ and $s'\in S_\infty$, we have
  $\mathbb{P}(s\to^n s')=0$ for all $n\in\mathbb{N}_0$ (also for the
  approximated chains $M_n$) whence
  \[\mu f(s)=\mu f^{(0)}(s)\text{ and }\mu f_n(s)=\mu f_n^{(0)}(s)\]
  for all $s\in S_?$.

  But, since $f_n^{(0)}\to f^{(0)}$ and all $f_n^{(0)}$ and $f^{(0)}$
  are power-contractions, we get the desired result
  \[\mu f(s)=\mu f^{(0)}(s)=\lim_{n\to\infty}\mu f_n^{(0)}(s)=\lim_{n\to\infty}\mu f_n(s).\]
  \qed
\end{proof}

\begin{toappendix}
    This now directly gives sampling-continuity for SSGs.
\end{toappendix}

We proved such a result already in~\cite{BGKPW:CAFDM},
yet some extra
work has to be done to adapt to the new setting, allowing for approximated rewards and
infinite values.
The idea is that states with zero or infinite value are determined by
the underlying structure of the Markov chain (in terms of
Definition~\ref{def:sampling}), whereas  on the \enquote*{non-trivial}
states, the Bellman operator enjoys a power contraction property
ensuring convergence to the (unique) fixpoint.

Using this result for Markov chains, it is easy to generalize to SSGs
using~\eqref{eq:value-equiv}.

\begin{corollaryrep}[Sampling-continuity of least-fixpoint operator]\label{cor:lfp-samp-cont}
    Given a sampling $\mathcal{G}=(G_n)$ of an SSG $G$,
    we have
    $\mu f_G=\lim_{n\to\infty}\mu f_{G_n}$.
\end{corollaryrep}
\begin{proof}
  This follows directly from~\Cref{lem:mp-sampled-lfp-convergence} and
  optimality of memoryless policies (see~\eqref{eq:value-equiv}):
    \begin{align*}
        \mu f_G&=\min_{\pi_{\min}}\max_{\pi_{\max}}\mu f^{(\pi_{\min},\pi_{\max})}\\
        &=\min_{\pi_{\min}}\max_{\pi_{\max}}\lim_{n\to\infty}\mu f_n^{(\pi_{\min},\pi_{\max})}\\
        &=\lim_{n\to\infty}\min_{\pi_{\min}}\max_{\pi_{\max}}\mu f_n^{(\pi_{\min},\pi_{\max})}\\
        &=\lim_{n\to\infty}\mu f_n
    \end{align*}
    using that only finitely many (memoryless) policies $\pi_{\min}\colon S_{\min}\setminus F\to A$ and
    $\pi_{\max}\colon S_{\max}\setminus F\to A$ exist.
    \qed
\end{proof}

We continue by showing Condition~\eqref{eq:conv-condition} for
MDPs and derive the following.
\begin{theoremrep}
  \label{th:mdp}
    Given a sampling $\mathcal{M}=(M_n)$ of an MDP $M$, we have\linebreak
    $\lim_{n\to\infty}\mu(\sup_{k\geq n}f_{M_k})=\mu f_M$.
    As a consequence, for any {\reg} Mann scheme $\mathcal{S}$, the sequence $(x_n)$
    generated by $\mathcal{S}$ on $(f_{M_n})$ converges to $\mu f_M$.
\end{theoremrep}
\begin{proofsketch}
  The main idea used to prove this statement is that the supremum of
  Bellman operators of arbitrary many MDPs of the same underlying
  structure (in the sense of Definition~\ref{def:sampling}) is itself a Bellman
  operator of an MDP with the \enquote*{same} structure (although with
  potentially many more actions) as long as all transition
  probabilities and rewards are close enough.  Using this, 
  Condition~\eqref{eq:conv-condition} can be shown to break down to
  continuity of the least fixpoint operator for a \enquote*{sampling}
  sequence.
  \qed
\end{proofsketch}
\begin{proof}
  Let $\tilde{f}_n=\sup_{k \geq n}f_n$ and let us show that $\mu\tilde{f}_n\to\mu f$.

    We write $T_{\min}=\min\{T_a(s,s')\mid s,s'\in S,a\in A(s),T_a(s,s')>0\}$ and
    $T_{\max}=\max\{T_a(s,s')\mid s,s'\in S,a\in A(s),T_a(s,s')<1\}$ for the smallest
    and largest non-trivial transition probability in $M$, and let
    $\varepsilon\in(0,\min(T_{\min},1-T_{\max}))$.

    We know that, for some $N\in\mathbb{N}$, we have $\|T_n-T\|<\varepsilon/|S|$ and
    $\|R_n-R\|<\varepsilon$ for all $n>N$.

    We now construct a new MDP $M_\varepsilon=(S,A_\varepsilon,T_\varepsilon,R_\varepsilon)$ as
    follows:
    For $s\in S$, we define
    \begin{eqnarray*}
      A_\varepsilon(s) & =& \{a_{s'}\mid a\in
      A(s),|\supp{T_a(s,\cdot)}|>1,s'\in\supp{T_a(s,\cdot)}\cup\mbox{}\\
      && \qquad \{a\mid
        a\in A(s),|\supp{T_a(s,\cdot)}|\leq 1\}
    \end{eqnarray*}
    For $a\in A(s)$ with $|\supp{T_a(s,\cdot)}|\leq 1$, we define
    \begin{eqnarray*}
      (T_\varepsilon)_a(s,s') &=& \begin{cases}
        \min(1,T_a(s,s')+\varepsilon)&\text{if }T_a(s,s')>0\\
        0&\text{otherwise}
      \end{cases} \\
      (R_\varepsilon)_a(s) &=& \begin{cases}
        R_a(s)&\text{if }R_a(s)=0\\
        R_a(s)+\varepsilon&\text{otherwise.}
      \end{cases}
  \end{eqnarray*}
    For $a\in A(s)$ with $|\supp{T_a(s,\cdot)}|>1$ and $t\in\supp{T_a(s,\cdot)}$, on the other hand,
    we define
    \[(T_\varepsilon)_{a_{t}}(s,s')=T_a(s,s')+\begin{cases}
        \varepsilon+(1-T_a(s))&\text{if }s'=t,\\
        -\nicefrac{\varepsilon}{|\supp{T_a(s,\cdot)}|-1}&\text{if }s'\in\supp{T_a(s,\cdot)}\setminus\{t\},\\
        0&\text{otherwise,}
    \end{cases}\]
    and
    \[(R_\varepsilon)_(a_t)(s)=\begin{cases}
        R_a(s)&\text{if }R_a(s)=0\\
        R_a(s)+\varepsilon&\text{otherwise.}
    \end{cases}.\]
    Intuitively, for every action yielding a non-trivial distribution, we add actions that each
    prefer a single possible output.
    Note that this definition still yields a well-defined MDP and, indeed, an MDP that has the
    same underlying structure in the sense of Definition~\ref{def:sampling} whence $(M_\varepsilon)$ can be
    understood as converging to $M$ in the sampling sense.
    In particular, we have by~\Cref{cor:lfp-samp-cont} that $\mu f_{M_\varepsilon}\to\mu f$ for
    $\varepsilon\to 0$.

    On the other hand, let $n>N,v\in\mathbb{R}_+^S,$ and $s\in S$ be arbitrary but fixed.
    We write $a^*$ for the action maximizing $(R_n)_a(s)+\sum_{s'\in S}(T_n)_a(s,s')v(s')$.
    If $\supp{T_a(s,\cdot)}=\emptyset$, we obviously have
    \[f_n(v)(s)=(R_n)_a(s)\leq R_a(s)+\varepsilon=(R_\varepsilon)_a(s)\leq f_\varepsilon(v)(s).\]
    Similarly, if $\supp{T_a(s,\cdot)}=\{s'\}$, we have
    \begin{align*}
      f_n(v)(s) & =(R_n)_a(s)+(T_n)_a(s,s')v(s') \\
                & \leq R_a(s)+\varepsilon+\min(1,T_a(s,s')+\varepsilon)v(s')\\
                & =(R_\varepsilon)_a(s)+(T_\varepsilon)_a(s,s')v(s')\\
                & \leq f_\varepsilon(v)(s).
    \end{align*}
    Finally, if $|\supp{T_a(s,\cdot)}|>1$ and $s^*$ maximizes $v(s^*)$, we have
    \begin{eqnarray*}
        f_n(v)(s)&=&(R_n)_{a^*}(s)+\sum_{s'\in S}(T_n)_{a^*}(s,s')v(s')\\
        &\leq&(R_{a^*}(s)+\varepsilon)+\sum_{s'\in
          S}(T_\varepsilon)_{a^*_{s^*}}(s,s')v(s') + \mbox{} \\
        && \qquad \sum_{s'\in S}((T_n)_{a^*}(s,s')-(T_\varepsilon)_{a^*_{s^*}}(s,s'))v(s')\\
        &\leq& f_{\varepsilon}(v)(s)+\sum_{s'\in
          S\setminus\{s^*\}}((T_n)_{a^*}(s,s')-(T_\varepsilon)_{a^*_{s^*}}(s,s'))v(s')
        + \mbox{} \\
        && \qquad ((T_n)_{a^*}(s,s^*)-(T_\varepsilon)_{a^*_{s^*}}(s,s^*))v(s^*)\\
        &\leq& f_{\varepsilon}(v)(s)+v(s^*)\sum_{s'\in S}((T_n)_{a^*}(s,s')-(T_\varepsilon)_{a^*_{s^*}}(s,s'))\\
        &\leq& f_{\varepsilon}(v)(s)+v(s^*)((T_n)_{a^*}(s)-(T_\varepsilon)_{a^*_{s^*}}(s))\\
        &\leq& f_{\varepsilon}(v)(s),
    \end{eqnarray*}
    where we make use of
    $(T_n)_{a^*}(s,s')\geq (T_\varepsilon)_{a^*_{s^*}}(s,s')=T_a(s,s')-\nicefrac{\varepsilon}{|\supp{T_a(s,\cdot)}|-1}$
    for all $s'\in\supp{T_a(s,\cdot)}\setminus\{s^*\}$ and
    $(T_\varepsilon)_{a^*_{s^*}}(s)=1\geq(T_n)_{a^*}(s)$.

    Since $n>N,v,$ and $s$ were arbitrary, this shows $f_n\leq f_{\varepsilon}$ for all $n>N$,
    i.e., $\tilde{f}_n=\sup_{k>n}f_k\leq f_\varepsilon$ and thus
    $\mu\tilde{f}_n\leq\mu f_\varepsilon$ for all $n>N$.
    All in all, this yields the desired result
    \[\mu f\leq\lim_{n\to\infty}\mu\tilde{f}_n\leq\lim_{\varepsilon\to0}\mu f_\varepsilon=\mu f.\]
    \qed
\end{proof}

Using the validity of Condition~\eqref{eq:conv-condition} for MDPs, we
can derive the same result for SSGs using~\eqref{eq:value-equiv} and
an optimal min-player strategy.

\begin{corollaryrep}
    Given a sampling $\mathcal{G}=(G_n)$ of an SSG $G$, we have
    that \linebreak $\lim_{n\to\infty}\mu(\sup_{k\geq n}f_{G_k})=\mu f_G$.
    As a consequence, for any {\reg} Mann scheme $\mathcal{S}$, the
    sequence $(x_n)$ generated by $\mathcal{S}$ on $(f_{G_n})$
    converges to $\mu f_G$.
\end{corollaryrep}
\begin{proof}
    Let $\pi\colon S_{\min}\setminus F\to A$ denote the optimal policy for the $\min$-player, let
    $f^\pi$ and $f_n^\pi$ denote the Bellman operators of the corresponding MDP and $(x_n^\pi)$
    the sequence generated by the corresponding iteration using $(f_n^\pi)$ instead of $(f_n)$.
    In particular, we have $\mu f=\mu f^\pi$ and $f_n\leq f_n^\pi$ for all $n\in\mathbb{N}_0$
    whence also $x_n\leq x_n^\pi$ for all $n\in\mathbb{N}_0$.
    But then, using convergence for MDPs (\Cref{th:mdp}), we immediately get
    \[\limsup_{n\to\infty}x_n\leq\lim_{n\to\infty}x_n^\pi=\mu f^\pi=\mu f,\]
    concluding the proof.
    \qed
\end{proof}

\begin{remark}
  This result can be used to show convergence for a wider range
  of \enquote*{sampled} piecewise linear functions by making use of
  the fact that convergence of {\reg} generalized schemes for a
  sampled SSG $(f_{M_n})$ directly implies convergence for the
  sequences $(f_{M_n}^k)_n$ for any fixed $k>1$ ($k$-step Bellman
  operators).
  The appendix reports additional results, showing how
  this can be used to, among others, directly show
  convergence for state-action-value Bellman operators, where rewards are given to transitions.
  \cut{\[g_M(q)(s,a)=R_a(s)+\sum_{s'\in S_{\max}}T_a(s,s')\max_{a'\in A(s')}q(s',a')+\sum_{s'\in S_{\min}}T_a(s,s')\min_{a'\in A(s')}q(s',a').\]}
\end{remark}

\begin{toappendix}
  To conclude this chapter, let us look at a further generalisation of
  this result.  To do so, we notice the two following simple tricks,
  showing convergence for multi-step Bellman operators and for closed
  subfunctions.
    \begin{lemma}\label{lem:multi-step}
        Given $k\in\N$, a sampling $(\mathcal{G}_n)$ of an SSG $\mathcal{G}$, and a {\reg} generalized dampened Mann
        scheme $\mathcal{S}$,
        the iteration $(x_n)$ generated by $\mathcal{S}$ on $(f_{\mathcal{G}_n}^k)_n$ converges to
        $\mu f_{\mathcal{G}}$.
    \end{lemma}
    \begin{proof}
        Let $k\in\N$, a sampling $(\mathcal{G}_n)$ of the SSG $\mathcal{G}$, and a {\reg} generalized dampened Mann
        scheme $\mathcal{S}=((\alpha_n),(\beta_n))$ be given.

        We notice that also $(\tilde{\mathcal{G}}_n)$ defined as
        \[\tilde{\mathcal{G}}_n=\mathcal{G}_{\lfloor n/k\rfloor}\]
        is still a sampling of $G$.

        Defining the scheme $\tilde{\mathcal{S}}=((\tilde{\alpha}_n),(\tilde{\beta}_n))$ as
        \[(\tilde{\alpha}_n,\tilde{\beta}_n)=\begin{cases}
            (\alpha_{(n-1)/k},\beta_{(n-1)/k})&\text{if }k\mid n-1\\
            (0,0)&\text{otherwise,}
        \end{cases}\]
        we notice that $\tilde{\mathcal{S}}$ is still {\reg} (choosing
        $\tilde{m}_\ell=k\cdot m_{\ell}$).

        As such, we have that the sequence $(\tilde{x}_n)$ generated by $\tilde{\mathcal{S}}$ on
        $(f_{\tilde{\mathcal{G}}_n})$ converges to $\mu f_{\mathcal{G}}$.
        But, by construction, we also have $x_n=\tilde{x}_{n\cdot k}$, showing that also $x_n\to\mu f_{\mathcal{G}}$.
        \qed
    \end{proof}

    For a function $f\colon \Rp^d\to\Rp^d$, we define its dependency
    graph as $(\{1,\dots,d\},\to)$ with edges / relation $\to$ given by
    \[i\to j\quad\text{iff}\quad\exists x,y\in\Rp^d\colon f(x)(j)\neq f(y)(j)\Rightarrow\exists k\neq i\colon x(i)\neq y(i).\]
    Now, for $I\subseteq\{1,\dots,d\}$, we write $\hat{x}(i)=x(i)$ for all $i\in I$ and
    $\hat{x}(i)=0$ otherwise, and define $f^I\colon \Rp^I\to\Rp^I$ as
    \[f^I(x)=(f(\hat{x})(i))_{i\in I}.\]

    Then, if $I\subseteq\{1,\dots,d\}$ is an \emph{independent} subset of components,
    meaning that for all $i\in I$ and $j\in\{1,\dots,d\}$ with $j\to i$, we have $j\in I$, then
    we have for any $x\in X$ that
    \[f^I((x(i))_{i\in I})=(f(x)(i))_{i\in I}.\]
    In particular, we have that $\mu f^I=(\mu f(i))_{i\in I}$.

    Furthermore, consider a (generalized) dampened Mann scheme $\mathcal{S}$ generating an iteration
    $(x_n)$ on a function sequence $(f_n)$ with a shared dependency graph.
    Then, for any independent subset $I$ of components, we have that the iteration $(x_n^I)_n$
    generated by the scheme $\mathcal{S}^I$ (which is just the parameter sequences of $\mathcal{S}$
    restricted to the components $I$) on $(f_n^I)_n$ fulfills $x_n^I=(x_n(i))_{i\in I}$, showing
    that convergence for the original function sequence naturally implies convergence for
    restrictions to independent subsets of components.

    These two ideas can now be combined to show convergence of {\reg} generalized Mann schemes for
    even greater classes of piecewise linear functions by identifying them with a restriction of the
    multi-step Bellman operator of an SSG.

    In order to demonstrate this, we consider the state-action Bellman
    operator $g_{\mathcal{G}}\colon \Rp^{S\times A}\to\Rp^{S\times A}$
    of an SSG $\mathcal{G}=(S,A,T,R,S_{\max},S_{\min})$:
    \begin{eqnarray*}
      g_{\mathcal{G}}(q)(s,a) &=& R_a(s)+\sum_{s'\in
        S_{\max}}T_a(s,s')\max_{a'\in A(s')}q(s',a')+\mbox{}\\
      && \qquad \sum_{s'\in
          S_{\min}}T_a(s,s')\min_{a'\in A(s')}q(s',a')
    \end{eqnarray*}

    This function is obviously not of the form of the state Bellman operator $f_{\mathcal{G}}$ used
    above, so the results above do not directly apply, despite using the \enquote*{same components}.
    Yet, with the two ideas above, showing convergence becomes rather straight-forward.

    \begin{corollary}
        Given a sampling $\mathcal{G}=(G_n)$ of an SSG $G$ and a {\reg} Mann scheme $\mathcal{S}$,
        the sequence $(x_n)$ generated by $\mathcal{S}$ on $(g_{G_n})$ converges to $\mu g_G$.
    \end{corollary}
    \begin{proof}
        For $G=(S,A,T,R,S_{\max},S_{\min})$ (and, analogously for $G_n$), we now define an SSG
        $\tilde{G}=(\tilde{S},\tilde{A},\tilde{T},\tilde{R},\tilde{S}_{\max},\tilde{S}_{\min})$ by
        splitting the probabilistic and non-deterministic part as follows:
        \begin{itemize}
            \item $\tilde{S}=S\cup(S\times A)$, i.e., all states and all state-action pairs
            \item $\tilde{A}(s)=A(s)$ for $s\in S$ and $\tilde{A}(s,a)=\{*\}$ for $(s,a)\in S\times A$
            \item $\tilde{T}$ is defined as follows (everywhere else $0$)
            \begin{itemize}
                \item $\tilde{T}_a(s,(s,a))=1$, i.e., choosing action $a$ in state $s$ deterministically yields the state-action pair $(s,a)$
                \item $\tilde{T}_*((s,a),s')=T_a(s,s')$, i.e., from state-action pair $(s,a)$, the next state is randomly chosen as when choosing action $a$ in state $s$ of the original SSG
            \end{itemize}
            \item $\tilde{R}$ is defined as follows
            \begin{itemize}
                \item $\tilde{R}_a(s,(s,a))=0$ for all $s\in S,a\in A$
                \item $\tilde{R}_*((s,a),s')=R(s,a)$ for all $s,s'\in S,a\in A$, i.e., the reward is granted when leaving the state-action pair
            \end{itemize}
            \item $\tilde{S}_{\min}=S_{\min}\cup(S\times A),\tilde{S}_{\max}=S_{\max}$\footnote{This choice is arbitrary as the state-action pairs offer no non-deterministic choice.}
        \end{itemize}
        Note that this gives us
        \[f_{\tilde{G}}(v)(s)=\max_{a\in A(s)}v(s,a)\]
        on states $s\in S$ and
        \[f_{\tilde{G}}(v)(s,a)=R_a(s)+\sum_{s'\in S}T_a(s,s')v(s')\]
        on state-action pairs $(s,a)\in S\times A$.

        In particular, we have that
        \[f_{\tilde{G}}^2(v)(s,a)=R_a(s)+\sum_{s'\in S}T_a(s,s')\max_{a'\in A(s')}v(s',a')=g_G(v)(s,a)\]
        for all state-action pairs $(s,a)\in S\times A$ (and, similarly, $f_{\tilde{G}}^2(v)(s)=f_G(v)(s)$ for all $s\in S$).

        Now, without loss of generality, we can assume that the structure (in the sense of~\Cref{def:sampling})
        of all SSGs $G_n$ is equal to the one of $G$.
        It is easy to check that, in this case, all $f_{\tilde{G}_n}^2$ share their dependency graph
        and that $S\times A$ (and also $S$) is an independent subset of components.

        But then, by~\Cref{lem:multi-step}, we have that any {\reg} generalized Mann scheme applied
        to $(f_{\tilde{G}_n}^2)$ yields a sequence $(x_n)$ convergent to $\mu f_{\tilde{G}}$.
        In particular, restricting the iteration to the shared independent subset $S\times A$, this
        gives convergence also for the state-action Bellman operators $(g_{\mathcal{G}_n})$.
        \qed
    \end{proof}
\end{toappendix}

\section{Numerical Experiments}
\label{se:numerical}

In this section we illustrate a number of experiments 
comparing the convergence behaviour of (relaxed) Mann-Kleene schemes as
presented in~\cite{BGKPW:CAFDM} with schemes working with vanishing or
divergent learning rates as well as chaotic schemes. The aim of these experiments
is to provide some insights on how the generalisations to dampened Mann
iteration in the paper affect the convergence speed when applied
with a randomized strategy for chaotic updates.
In particular, we cannot expect an improvement in efficiency that could
potentially be expected using more informed heuristics when choosing
which components to update.

We tested all iteration schemes on the same set of 50 randomly
generated SSGs $G$ with 15 $\min$- and 15 $\max$-player states, at
most $5$ actions per state, each of them with a normalized value
$\|\mu f_G\|=1$ and for which the Kleene iteration on $f_G$
\enquote*{converges} in at most $10000$ steps (we assume convergence
when changes are below $10^{-8}$).

\begin{figure}[t]
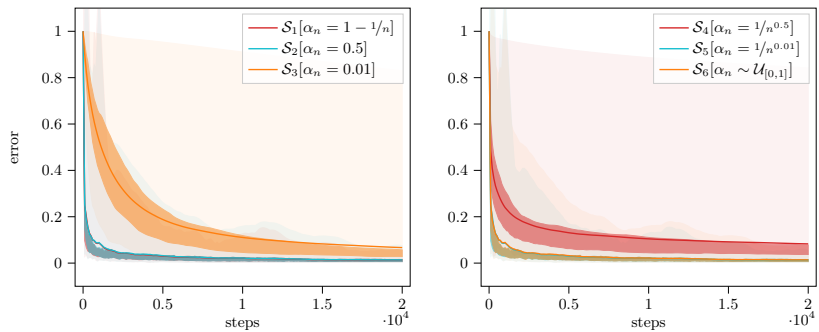

    \centering
    \begin{subfigure}[b]{.45\textwidth}
      \centering
\clearpage{}%


\clearpage{}%

    \end{subfigure}    
    \caption{Results for the experiments of non-chaotic iterations on
      50 randomly generated SSGs.  Lines denote the mean error,
      the half-opaque area the 25th to 75th percentile, the weakly
      visible area the minimum to maximum area. Results for
      $\mathcal{S}_1$ and $\mathcal{S}_4$ are only barely visible as
      they are almost identical to the ones from $\mathcal{S}_2$ and
      $\mathcal{S}_6$, respectively.}
    \label{fig:first-experiment}
\end{figure}

\paragraph{Vanishing and diverging learning rates.}
To test the influence of vanishing or diverging learning rates
on the convergence rate, we compared the performance of the
six iteration schemes shown in~\Cref{tab:first-experiment}.
The scheme $\mathcal{S}_1$ is a Mann-Kleene scheme, while $\mathcal{S}_2$ and $\mathcal{S}_3$  are relaxed
Mann-Kleene schemes.
Instead $\mathcal{S}_4$ and
$\mathcal{S}_5$ have vanishing learning rate and 
$\mathcal{S}_6$ chooses a random learning rate at each
step independently w.r.t. a uniform distribution, and thus the sequence of learning rates is 
almost surely non-converging.
Note that convergence for the first three schemes (at least for MDPs)
is covered by the results in~\cite{BGKPW:CAFDM} while the last three require the results in the present paper.

\begin{table}[b]
  \centering
  \begin{tabular}[h]{|c|c|c|c|c|c|c|}
    \hline
    $\mathcal{S}$&$\mathcal{S}_1$&$\mathcal{S}_2$&$\mathcal{S}_3$&$\mathcal{S}_4$&$\mathcal{S}_5$&$\mathcal{S}_6$\\
    \hline
    \mystrut$\alpha_n$&$1-\nicefrac{1}{n}$&$0.5$&$0.01$&$\nicefrac{1}{n^{0.5}}$&$\nicefrac{1}{n^{0.01}}$&$\sim\mathcal{U}_{[0,1]}$\\
    \mystrut$\beta_n$&$\nicefrac{1}{n}$&$\nicefrac{1}{n}$&$\nicefrac{1}{n}$&$\nicefrac{1}{n}$&$\nicefrac{1}{n}$&$\nicefrac{1}{n}$ \\ \hline
  \end{tabular}
  \medskip
  \caption{Dampened Mann Schemes used in the experiments.}
  \label{tab:first-experiment}
\end{table}

In the experiments, before each Mann iteration step, the SSGs $\mathcal{G}$ are sampled by performing $30$ model steps for randomly chosen state-action pairs.

The results of these experiments are
in~\Cref{fig:first-experiment}.  Observe that both vanishing and random
(diverging) learning rates yield similar results as (relaxed)
Mann-Kleene parameters.

Also, the schemes that keep having a higher learning rate (for longer)
perform better (higher constant learning rates or lower exponent
$\alpha$ for learning rates of the form $\nicefrac{1}{n^\alpha}$). This is
expected since we are dealing with faultless sampling of SSGs and the
advantage of lower learning rates does not lie in faster convergence
but in higher robustness to \enquote*{irregular} approximations
(see~\Cref{ex:weigthed-convergence}).

\paragraph{Chaotic iteration.}

\begin{figure}[t]
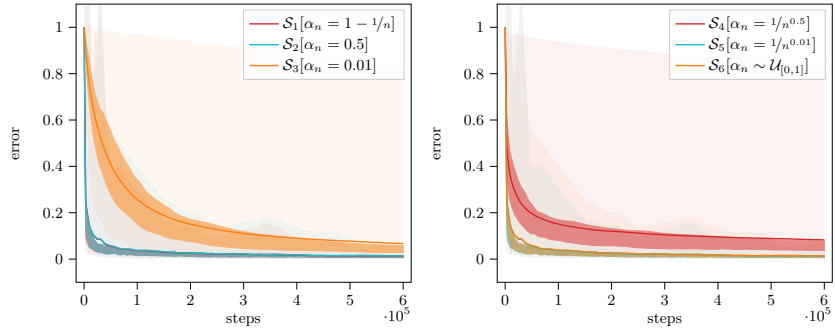

    \centering
    \begin{subfigure}[b]{.45\textwidth}
      \centering
\clearpage{}%


\clearpage{}%

    \end{subfigure}
    
    \caption{Results for the experiments of chaotic iterations on 50 randomly generated SSGs.
    Again, the results for $\mathcal{S}_1$ and
        $\mathcal{S}_4$ are only barely visible as they are almost identical to the ones from
        $\mathcal{S}_2$ and $\mathcal{S}_6$, respectively.}
      \label{fig:second-experiment}
\end{figure}
To test the influence of chaotic updates on the convergence rate,
we repeated the experiments with chaotic variants of the above
iteration schemes.
As outline at the end of \S\ref{ss:gen-mann}, at each step, the chaotic iteration performs an
update for one single (independently) randomly chosen state,
updating the learning rate (and dampening factor) only for this
state.
Thus, to have a \enquote*{full sweep} over the states, the
chaotic variants need (at least) $30$ steps (while the
non-chaotic variants above perform $30$ updates in one step).
To have a fair comparison to the non-chaotic variants, the
sampling $\mathcal{G}$ only simulates a single model step
for a randomly chosen state-action pair before each iteration
step.
In total, after $30$ steps, the chaotic schemes thus have
the same information and the same number of updates (thus, the
same amount of computation) as a single step of the non-chaotic
variants, while allowing for intermediate results.

The results of these experiments are in~\Cref{fig:second-experiment}.
 It can be seen directly that the schemes all have almost the
same convergence speed as the non-chaotic variants (keeping in mind
that the non-chaotic variants perform $30$ updates per step), even
if we use  a random non-informed choice of components to update in
each step.
As a consequence, the experiments suggest that performing
a chaotic iteration (randomly) comes at no additional cost
while allowing for intermediate results even when the
computation with non-chaotic variants would be infeasible.
It should further be noted that, in most practical use
cases, the components to update might be chosen in a more
informed manner, e.g., by taking into account which components
had been updated, which might result in faster convergence.

\section{Conclusion}
\label{se:conclusions}
Mann iteration is a classical technique for generating a sequence
which, under suitable hypotheses, converges to a fixpoint of a
continuous function starting from any initial
state~\cite{b:iterative-approximation-fixed-points}. Dampened versions
are considered in~\cite{KIM200551,Yao2008StrongCO} where convergence
to a fixpoint is shown in
a uniformly smooth Banach space and Hilbert spaces, under suitable assumptions on the parameter sequences.
In~\cite{BGKPW:CAFDM}, we proposed the use of dampened Mann iteration for
dealing with functions which can only be approximated. This work
sensibly deviates from~\cite{KIM200551,Yao2008StrongCO} as the
intended applications, for approximating fixpoints of
functions arising from quantitative models like MDPs and SSGs,
lead to focus on the case of
monotone non-expansive functions in (finite-dimensional) Banach spaces
with the supremum norm, which are neither uniformly smooth nor Hilbert
spaces.

Our paper takes the same setting as~\cite{BGKPW:CAFDM} and generalizes
its results. The generalisations include the possibility of having
learning rates converging to $0$ or not converge at all, the
possibility of varying parameters across dimensions, thus enabling
chaotic iteration and the proof of convergence of these schemes for
simple stochastic games approximated via sampling.

We provided some numerical results, showing that the chaotic
iteration (with random choice of updated component) yields almost
identical results in terms of convergence of standard dampened Mann
iteration, while giving more flexibility, e.g., opening the way to
parallel and distributed implementations.
We plan a further exploration of effective chaotic iteration strategies, including data-driven approaches to select the components to update at each step based on statistical properties of the model, possibly tuning the learning rate and dampening parameters based on observed convergence behaviour during execution.

We furthermore plan to investigate applications of our results to the
computation of behavioural metrics, in particular probabilistic
bisimilarity distances. In
fact,~\cite{bblmtb:prob-bisim-distance-automata-journal} shows how to
reduce this problem to the solution of a simple stochastic
game. Computing behavioural metrics in the approximated setting is
also
discussed~\cite{kt:approximate-bisimulation-minimization,FKPB:RPBLMC}
where the authors observe that a small perturbation of the transition
probabilities can drastically change the distance apparently hindering
the possibility of properly approximating the distance when such
probabilities can be only estimated.  Our results suggest that,
instead, one can use dampened Mann iteration for obtaining
increasingly better approximations of the distance.

We are also interested in considering the problem of approximating the least fixpoint in scenarios,
where the functions do not necessarily converge, but converge in the
limit-average, providing a close match with reinforcement learning
algorithms such as Q-learning \cite{WD:QL}. For this one might be able
to draw inspiration from theory of stochastic
approximation~\cite{bt:neuro-dynamic-programming,b:random-iterative-models,kc:stochastic-approximation-methods},
going back to~\cite{rm:stochastic-approximation} that deals with the
case where a function contains an error term with expected value zero.

\bibliographystyle{plain}
\bibliography{references}

\end{document}